\documentclass[11pt]{article}
\pdfoutput=1
\usepackage[totalwidth=480pt, totalheight=680pt]{geometry}
\usepackage[latin1]{inputenc}
\usepackage{hyperref}
\usepackage{amsmath}
\usepackage{amssymb}
\usepackage{amsthm}
\usepackage{latexsym}
\usepackage{bbm}
\usepackage[above,below]{placeins}
\usepackage[title,titletoc]{appendix}
\usepackage{ifpdf}
\ifpdf
    \usepackage[pdftex]{graphicx}
\else
    \usepackage{graphicx}
\fi

\newcommand{\ud}{\mathrm{d}}

\newtheorem{theorem}{Theorem}
\newtheorem{proposition}{Proposition}

\ifpdf
\usepackage[pdftex]{graphicx}
\else
\usepackage{graphicx}
\fi

\graphicspath{{plots/}}

\begin{document}

\ifpdf
\DeclareGraphicsExtensions{.pdf, .jpg, .tif}
\else
\DeclareGraphicsExtensions{.eps, .jpg}
\fi

\thispagestyle{empty}


\vspace*{3cm}

\begin{center}
\begin{Large}
\textbf{A Dynamic Model
for
Credit Index Derivatives}
\end{Large}

\vspace{15mm} {\bf Louis Paulot}

\vspace{7mm}
\emph{Sophis
\\
24--26 place de la Madeleine, 75008 Paris, France}

\vspace{5mm}
{\ttfamily louis.paulot@sophis.net}

\vspace{2cm}

November 2009

\end{center}

\vspace{3cm}
\hrule
\begin{abstract}
We present a new model for credit index derivatives, in the \emph{top-down} approach. This model has a dynamic loss intensity process with volatility and jumps and can include counterparty risk. It handles CDS, CDO tranches, N\textsuperscript{th} -to-default and index swaptions. Using properties of affine models, we derive closed formulas for the pricing of index CDS, CDO tranches and N\textsuperscript{th}-to-default. For index swaptions, we give an exact pricing and an approximate faster method. We finally show calibration results on 2009 market data.
\end{abstract}
\hrule

\vspace{\stretch{1}}

\pagebreak

\tableofcontents

\pagebreak

\parskip=4pt

\section{Introduction}

Credit derivatives allow to transfer some credit risk: an investor can get or sell protection on a debt issuer or on a basket of different names. Most popular credit basket derivatives are Credit Default Swaps (CDS), where one gets protection on a whole basket against some premium paid on a regular basis. For CDO tranches, the protection does not cover losses below some trigger or above a second one. In N\textsuperscript{th}-to-default, the protection concerns precisely the N\textsuperscript{th} default on the basket. An other kind of derivatives are CDS options (swaptions): one gets the right to enter at a future date in a CDS at a given premium rate. Depending on the market price of protection at the option maturity, the option is exercised or not.

There are different risks in credit derivatives. The risk of default should be the more obvious. A second kind of risk are movements in the market price of CDS. When the derivative is written on a basket, risks comes from default risk and spread risk on underlyings and also depends on correlation between names.

There is a first distinction between credit models. Some models are \emph{static}: the intensity of default of a name has no stochastic behaviour with respect to time evolution. It is sufficient to price CDS or CDO tranches, but such models cannot handle options: the spread risk is not taken into account. \emph{Dynamic} models incorporate this risk, generally by making the default intensity stochastic.

There is an other splitting between credit basket models. In \emph{bottom-up} models, credit basket derivatives prices and risks are reconstructed from individual constituents and their correlations \cite{li2000default,duffie2001risk}. An alternative approach describes directly the risks on the basket \cite{errais2009affine,schonbucher2006portfolio,brigo2007calibration,arnsdorf2009BSLP}. It is sufficient if one wants to hedge derivatives on a credit index with the index itself or other derivatives and not with individual components. If one is concerned with individual names risk, a second step is necessary, such as the \emph{random thinning} \cite{giesecke2009topdown}.

We introduce a model which belongs to the second category: we model directly defaults on the basket. Our model is also dynamic: the intensity of default is a dynamic, stochastic variable. More precisely, the intensity of default of the basket is modeled by a Cox-Ingersoll-Ross process with jumps. A diffusion term describes short movements of the credit index spread. Jumps intend to take into account possible crises where the perceived credit quality of the basket decreases by a large amount. After crises, the intensity can relax to a lower value. We also introduce a probability for the default of the whole basket. Our model can handle counterparty risk: the counterparty can be correlated to the basket or can even be one of the basket names.

We define first a model on an infinite pool of names where default events occur according to a stochastic default intensity, without any bound on the total number of defaults. In a second time we use this infinite pool model to model defaults on a finite basket. Using properties of \emph{affine models} \cite{duffie2000transform}, we get closed formulas for the (conditional) expected loss or tranches losses. Usual credit derivatives can be priced exactly: CDS, CDO tranches, N\textsuperscript{th}-to-default, index CDS options. It takes reasonable time for most instruments. However for swaptions the exact pricing can be slower so we introduce some approximations which reduces highly the computation time.

We describe our model in section~\ref{sec-model} and the pricing of usual credit derivatives in section~\ref{sec-pricing}. We finally show numerical results of calibration on market data in section~\ref{sec-calibration}.

\section{Model}
\label{sec-model}

\subsection{Infinite pool}

\subsubsection{Description of the model}
\label{sec:model}

The credit portfolio in modeled in a \emph{top-down} approach: the portfolio loss is considered directly. Probabilities are considered to be risk-neutral probabilities. We consider first a portfolio of infinite size. Its loss $\widetilde{L}_t$ is a stochastic process with jumps $l\, \ud J$ which occur with an intensity $\lambda_t$ and size $l$:
\begin{equation*}
\ud \widetilde{L}_t = l \, \ud J
\end{equation*}
where $\ud J$ are jumps of size 1.
The recovery rates $R$ of defaults are independent and identically distributed. We denote by $\ud \nu_J(l)$ the distribution of the jump size $l=1-R$.
In parallel, we have a stochastic process $\widetilde{N}_t$ counting the number of default. It is similar to $\widetilde{L}_t$ but with zero recovery rate:
\begin{equation*}
\ud \widetilde{N}_t = \ud J
\rlap{\ .}
\end{equation*}

We want to build a model for the default intensity process with the following features:
\begin{itemize}
\item a diffusion process describing small fluctuations of the CDS index spread;
\item positive jumps of the intensity which can produce episodes of crisis with many defaults on some timescale;
\item a mean reversion such that crises can relax.
\end{itemize}

Mathematically, the intensity $\lambda_t$ will follow a stochastic process like
\begin{equation*}
\ud \lambda_t = \kappa (\lambda_\infty - \lambda_t) \ud t + \sigma(\lambda_t) \ud W + \ud X
\end{equation*}
where $\ud W$ is a Brownian motion and $\ud X$ a jump process with intensity $\gamma$.

It appears that this model has analytic solutions if we make some further assumptions on the diffusion and jump process:
\begin{itemize}
\item the volatility is proportional to the square root of $\lambda$: $\sigma(\lambda_t) = \sigma \sqrt{\lambda_t}$, as in a Cox-Ingersoll-Ross process;
\item the jump size are independent and identically distributed along a Gamma distribution of parameters $n$ and $\theta$ where $n$ is an integer which describes the shape of the distribution and the mean size of jumps is $(n+1)\theta$:
\begin{equation}
\ud \nu_X(x) = \frac{x^n e^{-x/\theta}}{n! \theta^{n+1}} \mathbbm{1}_{x > 0}
\rlap{\ ;}
\label{dNuX}
\end{equation}
\item all processes are independent of each other (except the fact that $\lambda_t$ is the intensity of jumps $\ud J$).
\end{itemize}
The justification of this set of assumption is computational tractability. This choice encompasses the features we want to model without restricting the model too much. In fact one feature we might want to include is missing: here jumps in the hazard rate are purely exogenous, there is no default contagion, as in \cite{errais2009affine}. Including a jump term proportional to $\ud \widetilde{L}_t$ for the intensity process requires numerical integration of differential equations, which is tractable but much slower than analytical integration. A few tests do not indicate that including default contagion changes the loss distributions and therefore the CDO prices.

The default intensity process $\lambda_t$ is thus supposed to follow a Cox-Ingersoll-Ross process with additional jumps with intensity $\gamma$ and size distributed along a Gamma distribution (equation~\eqref{dNuX}):
\begin{equation*}
\ud \lambda_t = \kappa (\lambda_\infty - \lambda_t) \ud t + \sigma \sqrt{\lambda_t} \ud W + \ud X
\rlap{\ .}
\end{equation*}
Parameters $\kappa$, $\lambda_\infty$, $\sigma$, $n$, $\gamma$ and $\theta$ can be taken to be constant or piecewise constant. The model parametrization contains in addition the initial default intensity $\lambda_0$.

Experience shows that calibration of the model without a probability of default of all names tends to incorporate it by setting a very high jump size in $\ud X$. We therefore introduce explicitly such a possibility. We introduce a variable $Q$ which jumps of one unit when all names default together, with intensity $\alpha \lambda_t + \beta$. $Q\neq0$ models a default of the whole portfolio. Parameters $\alpha$ and $\beta$ control the probability of such an event and its correlation to the short-term credit quality of the basket modeled by $\lambda$. $\alpha$ gives a contribution independent of the current state of the system whereas $\beta$ gives a higher probability of total default during crises.

In order to handle counterparty risk, we introduce finally a variable $R$ which jumps of one unit when the counterparty defaults. It jumps with probability $\xi$ when there is a jump in $\ud J$ and additionally with intensity $\zeta \lambda_t + \eta$ at any time. When the counterparty belongs to the set of index names, $\zeta$ and $\eta$ are set to zero and $\xi$ controls the probability that a default concerns the counterparty. When the counterparty is not one of the names, $\xi$ is set to zero, $\zeta$ and $\eta$ control the intensity of a default of the counterparty and its correlation with defaults in the basket.

The intensity process jump distribution $\ud \nu_X(x)$ given in \eqref{dNuX} can be more general than a pure gamma distribution by considering the sum of several gamma distributions with different parameters $\gamma^{(i)}$, $n^{(i)}$ and $\theta^{(i)}$. If $\theta$ is kept constant, one can for example take a distribution of the form $P(x) e^{-x/\theta}$ where $P$ is any polynomial. For clarity we write formulas for only one gamma distribution; the generalization is straightforward.

\subsubsection{Characteristic function}

\paragraph{Theorem}

\begin{theorem}
\label{thm}
For constant parameters, the characteristic function at date $t$ in the four variables $\widetilde{L}_T$, $\widetilde{N}_T$, $\lambda_T$ and $Q_T$ for complex variables $u$, $v$, $w$ and $x$ is
\begin{equation*}
\mathbb{E}_t \left( e^{u \widetilde{L}_T + v \widetilde{N}_T + w \lambda_T + x Q_T + y R_T } \right) = e^{A(t,T;u,v,w,y) + B(t,T;u,v,w,y) \lambda_t + u \widetilde{L}_t + v \widetilde{N}_t + x Q_t + y R_t}
\end{equation*}
with\footnote{To simplify notation we drop the explicit dependency of functions $A$, $B$, $B_\pm$, $\chi$ and $\psi$ in variables $t$, $T$, $u$, $v$, $w$, $x$ and $y$.}
\begin{equation*}
B = B_- + \frac{B_+ - B_-}{1-\frac{1}{\chi}}
\label{solB}
\end{equation*}
\begin{equation*}
B_\pm = \frac{\kappa \pm \sqrt{\kappa^2 + 2 \sigma^2 \psi}}{\sigma^2}
\end{equation*}
\begin{equation*}
\chi = \frac{w - B_-}{w - B_+} e^{-\sqrt{\kappa^2 + 2 \sigma^2 \psi } \, (T-t)}
\end{equation*}
\begin{equation*}
\psi =  1-e^v \phi_J(u) \big[ 1 - \xi (1-e^y) \big] + \alpha (1-e^x) + \zeta (1-e^y)
\end{equation*}
\begin{equation*}
\begin{array}{rcl}
A &=& \displaystyle \frac{1}{\frac{1}{2} \sigma^2 (B_+-B_-)} \times \sum_\pm \Bigg(
\begin{array}[t]{l}
\displaystyle \pm \left[ \kappa \lambda_\infty B_\pm + \gamma \left( \frac{1}{(1-\theta B_\pm)^{n+1}} -1 \right) \right] \ln\!\left( \frac{B-B_\pm}{w-B_\pm} \right)
\\
\displaystyle \mp \gamma \frac{1}{(1-\theta B_\pm)^{n+1}} \ln\!\left( \frac{1-\theta B}{1-\theta w} \right)
\\
\displaystyle \pm \gamma \sum_{k=1}^n \frac{1}{k} \frac{1}{(1-\theta B_\pm)^{n-k+1}} \left[\frac{1}{(1-\theta B)^k} - \frac{1}{(1-\theta w)^k}\right]\Bigg)
\end{array}
\\
&& \displaystyle - \beta (1-e^x) (T-t) - \eta (1-e^y) (T-t) 
\rlap{\ .}
\end{array}
\end{equation*}
$\phi_J(z)$ is the characteristic function of the individual loss size:
\begin{equation*}
\phi_J(z) = \int e^{z l}\ud \nu_J(l) \rlap{\ .}
\end{equation*}
\end{theorem}

The proof of this theorem and its generalization to piecewise constant parameters is described in the following.

\paragraph{Differential equations}

\begin{proposition}
The model we have described is an affine model in the three variables $\widetilde{L}_t$, $\widetilde{N}_t$ and $\lambda_t$: its characteristic function can be written in the form
\begin{multline}
\mathbb{E}_t \left( e^{u \widetilde{L}_T + v \widetilde{N}_T + w \lambda_T + x Q_T + y R_T} \right) = 
\\
e^{A(t,T;u,v,w,x,y) + B(t,T;u,v,w,x,y) \lambda_t + C(t,T;u,v,w,x,y) \widetilde{L}_t + D(t,T;u,v,w,x,y) \widetilde{N}_t + E(t,T;u,v,w,x,y) Q_t + F(t,T;u,v,w,x,y) R_t}
\label{affine}
\end{multline}
where $A$, $B$, $C$, $D$, $E$  and $F$ are complex functions of the real parameters $t$ and $T$ and the complex parameters $u$, $v$, $w$, $x$ and $y$ which satisfy differential equations \eqref{eqA}--\eqref{eqF} and boundary conditions \eqref{bcA}--\eqref{bcF}.
\end{proposition}

\begin{proof}
Let suppose that equation \eqref{affine} holds. The variation of $\mathbb{E}_t \left( e^{u \widetilde{L}_T + v \widetilde{N}_T +w \lambda_T + x Q_T + y R_T} \right)$ when $t$ goes to $t+\ud t$ can be written as
\begin{multline}
\ud \mathbb{E}_{t} \left( e^{u \widetilde{L}_T + v \widetilde{N}_T +w \lambda_T + x Q_T + y R_T} \right) = \Big[ \partial_t A \ud t
 + \partial_t B  \ud t \, \lambda_t + \left(e^{B \ud \lambda_t}-1\right)
+ \partial_t C  \ud t \,\widetilde{L}_t + \partial_t D  \ud t \, \widetilde{N}_t
\\
+ \partial_t E \ud t \, Q_t + \partial_t F \ud t \, R_t
+ \left(e^{C \ud \widetilde{L}_t + D \ud \widetilde{N}_t + F \ud Q_t}-1\right) + \left(e^{E \ud Q_t}-1\right) \Big]
e^{A + B \lambda_t + C \widetilde{L}_t + D \widetilde{N}_t + E Q_t + F R_t}
\label{affine-2}
\end{multline}
where
\begin{eqnarray*}
e^{B \ud \lambda_t}-1 &=& B \kappa (\lambda_\infty - \lambda_t) \ud t + B  \sigma \sqrt{\lambda_t} \ud W + \frac{1}{2} B^2 \sigma^2 \lambda_t \ud t + \left(e^{B \ud X}-1 \right)
\\
e^{C \ud \widetilde{L}_t+D \ud \widetilde{N}_t + F \ud R_t}-1 &=& e^{(C j + D) \ud J + F \ud R}-1
\rlap{\ .}
\end{eqnarray*}

On the other hand, we know from the trivial equation
\begin{equation*}
\mathbb{E}_t \left[ \mathbb{E}_{t+\ud t} \left( e^{u \widetilde{L}_T + v \widetilde{N}_T +w \lambda_T + x Q_T + y R_T} \right) \right] = \mathbb{E}_t \left( e^{u \widetilde{L}_T + v \widetilde{N}_T +w \lambda_T + x Q_T + y R_T} \right)
\end{equation*}
that
\begin{equation*}
\mathbb{E}_t \left[ \partial_t \mathbb{E}_{t} \left( e^{u \widetilde{L}_T + v \widetilde{N}_T +w \lambda_T + x Q_T + y R_T} \right) \right] = 0
\rlap{\ .}
\end{equation*}
Therefore the expected value of \eqref{affine-2} must vanish:
\begin{multline*}
0 = \Bigg[ \partial_t A \ud t
 + \partial_t B \ud t \,\lambda_t + B \kappa (\lambda_\infty - \lambda_t) \ud t + \frac{1}{2} B^2 \sigma^2 \lambda_t \ud t + \gamma \ud t \int \left(e^{B x}-1 \right) \ud \nu_X(x)
 \\
+ \partial_t C \ud t \,\widetilde{L}_t + \partial_t D \ud t \, \widetilde{N}_t 
 + \partial_t E \ud t \, Q_t + \partial_t F \ud t \, R_t
+ \lambda_t \ud t \left( e^D \int e^{C j} \ud \nu_J(j) \big[ \xi e^F + (1-\xi) \big] -1 \right)
\\
 + (\alpha \lambda_t + \beta) \ud t \left (e^E-1\right)+ (\zeta \lambda_t + \eta) \ud t \left(e^F-1\right) \Bigg]
e^{A + B \lambda_t + C \widetilde{L}_t + D \widetilde{N}_t + E Q_t + F R_t}
\rlap{\ .}
\end{multline*}
Forgetting the last exponential, this equation is affine in $\widetilde{L}_t$, $\widetilde{N}_t$, $\lambda_t$, $Q_t$ and $R_t$.
As this must be true for any value of stochastic variables $\widetilde{L}_t$, $\widetilde{N}_t$, $\lambda_t$, $Q_t$ and $R_t$, this equation holds if and only if the following equations are satisfied
\begin{eqnarray}
\partial_t A(t,T;u,v,w,x,y) &=& - \kappa \lambda_\infty B - \gamma \left( \phi_X(B) -1 \right) - \beta\left (e^E-1\right)  - \eta\left (e^F-1\right)
\label{eqA}
\\
\partial_t B(t,T;u,v,w,x,y) &=& \kappa B - \frac{1}{2} \sigma^2 B^2 - \left( e^D \phi_J(C) \big[ \xi e^F + (1-\xi) \big] -1 \right)
\nonumber\\
&& - \alpha \left(e^E-1\right) - \zeta \left(e^F-1\right)
\label{eqB}
\\
\partial_t C(t,T;u,v,w,x,y) &=& 0
\label{eqC}
\\
\partial_t D(t,T;u,v,w,x,y) &=& 0
\label{eqD}
\\
\partial_t E(t,T;u,v,w,x,y) &=& 0
\label{eqE}
\\
\partial_t F(t,T;u,v,w,x,y) &=& 0
\rlap{\ .}
\label{eqF}
\end{eqnarray}
with the boundary conditions at $t=T$
\begin{eqnarray}
A(T,T,u,v,w,y) &=& 0
\label{bcA}
\\
B(T,T,u,v,w,y) &=& w
\label{bcB}
\\
C(T,T,u,v,w,y) &=& u
\label{bcC}
\\
D(T,T,u,v,w,y) &=& v
\label{bcD}
\\
E(T,T,u,v,w,y) &=& x
\label{bcE}
\\
F(T,T,u,v,w,y) &=& y
\rlap{\ .}
\label{bcF}
\end{eqnarray}
We have used the following notations for the characteristic functions of $\ud \nu_J$ and $\ud \nu_X$
\begin{eqnarray*}
\phi_J(z) &=& \int e^{z l}\ud \nu_J(l) \\
\phi_X(z) &=& \int e^{z x}\ud \nu_X(x)
\rlap{\ .}
\end{eqnarray*}

As a consequence, equation \eqref{affine} holds if and only if these differential equations are satisfied with their boundary conditions.
\end{proof}

\paragraph{Analytical solution}

These equations can be solved in closed form, when parameters are piecewise constant.
We consider first the case where all parameters are constant, which completes the proof of theorem \ref{thm}.

The equations for $C$, $D$, $E$ and $F$ have the simple solutions
\begin{eqnarray*}
C(t,T;u,v,w,x,y) &=& u
\\
D(t,T;u,v,w,x,y) &=& v
\\
E(t,T;u,v,w,x,y) &=& x
\\
F(t,T;u,v,w,x,y) &=& y
\rlap{\ .}
\end{eqnarray*}

The equation for $B$ is a Riccati equation with constant coefficients:
\begin{equation}
\partial_t B(t,T;u,v,w,x,y) = - \frac{1}{2} \sigma^2 B^2 + \kappa B + \psi
\label{riccati}
\end{equation}
with
\begin{equation*}
\psi(u,v,x,y) =  1-e^v \phi_J(u) \big[ 1 - \xi (1-e^y) \big] + \alpha (1-e^x) + \zeta (1-e^y)
\rlap{\ .}
\end{equation*}

The standard way for solving this equation is:
\begin{enumerate}
\item find a particular (constant) solution $b$;
\item write $B = b + f$, which gives a Bernoulli equation;
\item change the variable to $g = \frac{1}{f}$ to get a linear equation.
\end{enumerate}

Instead of this, we give a more pedestrian but more direct derivation of the solution, which is more in the spirit of the solution of the equation for $A$. We rewrite equation \eqref{riccati} as
\begin{equation*}
\frac{\ud B(t,T;u,v,w,x,y)}{- \frac{1}{2} \sigma^2 B^2 + \kappa B + \psi} = \ud t
\end{equation*}
which can be integrated directly.

We find the simple poles decomposition after rewriting again this equation as
\begin{equation*}
\frac{\ud B}{\left(B - B_+\right) \left(B - B_-\right) } = -\frac{\sigma^2}{2} \ud t
\end{equation*}
where $B_\pm$ are in fact the constant solutions of the differential equation \eqref{riccati}:
\begin{equation*}
B_\pm = \frac{\kappa \pm \sqrt{\kappa^2 + 2 \sigma^2 \psi}}{\sigma^2}
\rlap{\ .}
\end{equation*}
With non zero $\kappa$ and $u$ on the imaginary axis $B_+$ and $B_-$ are distinct and the simple pole decomposition
\begin{eqnarray*}
\frac{\ud B}{B - B_+}-\frac{\ud B}{B - B_-} &=& -\frac{2(B_+ - B_-)}{\sigma^2} \ud t
\\
&=& - \sqrt{\kappa^2 + 2 \sigma^2 \psi} \ud t
\end{eqnarray*}
is integrated as
\begin{equation*}
\ln\!\left(\frac{B - B_+}{B - B_-}\right) = \ln\!\left(\frac{B_T - B_+}{B_T - B_-}\right) + \sqrt{\kappa^2 + 2 \sigma^2 \psi} (T-t)
\end{equation*}
where we simplify slightly the notations by writing $B=B(t,T;u,v,w,x,y)$ and $B_T=B(T,T,u,v,w,y)$.
The exponential of this equation is
\begin{equation*}
\frac{B - B_-}{B - B_+} = \frac{B_T - B_-}{B_T - B_+} e^{-\sqrt{\kappa^2 + 2 \sigma^2 \psi} (T-t)} = \chi(t,T;u,v,w,x,y)
\end{equation*}
where we have introduced the notation $\chi(t,T;u,v,w,x,y)$. It gives the following expression for $B$:
\begin{equation}
B = \frac{B_- - \chi B_+}{1-\chi} = B_- + \frac{B_+ - B_-}{1-\frac{1}{\chi}}
\label{solB}
\rlap{\ .}
\end{equation}
When $T-t$ goes to infinity $\chi$ goes to zero and $B$ goes to $B_-$.

We are left with the equation for $A$. For the moment, we set $\beta=0$ and $\eta=0$ (the last terms of equation~\eqref{eqA} will give simple contributions that we give at the end of this section). This equation will be solved as a function of $B$. Dividing equation \eqref{eqA} (without the last term) by equation \eqref{eqB} gives
\begin{equation*}
\frac{\ud A}{\ud B} = \frac{\kappa \lambda_\infty B + \gamma \left( \phi_X(B) -1 \right)}{ \frac{1}{2} \sigma^2 B^2 - \kappa B + \left( e^D \phi_J(C) -1 \right)}
\rlap{\ .}
\end{equation*}
The assumption we have made on the jump size distribution ensures that the right-hand side is a rational function. Equation \eqref{dNuX} gives the characteristic function of this distribution as
\begin{equation*}
\phi_X(x) = (1-\theta z)^{-(n+1)}
\rlap{\ .}
\end{equation*}
Therefore all poles of the rational functions are known: $B_+$ and $B_-$ are simple poles and there is a pole of order $n+1$ in
\begin{equation*}
B_\theta = \frac{1}{\theta}:
\end{equation*}
\begin{equation*}
\frac{\ud A}{\ud B} = \frac{\kappa \lambda_\infty B + \gamma \left( (1-\theta B)^{-(n+1)} -1 \right)}{ \frac{1}{2} \sigma^2 (B-B_+)(B-B_-)}
\rlap{\ .}
\end{equation*}

To write down the pole decomposition, the first step is to use
\begin{equation*}
\frac{1}{\frac{1}{2} \sigma^2 (B-B_+)(B-B_-)} = \frac{1}{\frac{1}{2} \sigma^2 (B_+-B_-)} \left( \frac{1}{B-B_+} - \frac{1}{B-B_-} \right)
\end{equation*} to separate the fraction in two parts:
\begin{equation}
\frac{\ud A}{\ud B} = \frac{1}{\frac{1}{2} \sigma^2 (B_+-B_-)} \left( \frac{F(B)}{B-B_+} - \frac{F(B)}{B-B_-} \right)
\label{twoparts}
\end{equation}
where
\begin{equation*}
F(B) = \kappa \lambda_\infty B - \gamma + \frac{\gamma ( - B_\theta)^{n+1}}{(B-B_\theta)^{n+1}}
\rlap{\ .}
\end{equation*}
The first two terms give a contribution to $\pm \frac{F(B)}{B-B_\pm}$
\begin{equation*}
\pm \kappa \lambda_\infty \pm \frac{\kappa \lambda_\infty B_\pm - \gamma}{B-B_\pm}
\end{equation*}
where the constant term will be simplified between $+$ and $-$ components.
The other poles of $\pm \frac{F(B)}{B-B_\pm}$ come from the term $\pm \frac{\gamma ( - B_\theta)^{n+1}}{(B-B_\theta)^{n+1}}$:
\begin{equation*}
\pm \frac{\gamma ( - B_\theta)^{n+1}}{(B-B_\theta)^{n+1}}
=
\pm \gamma ( - B_\theta)^{n+1} \left[ \frac{k_\pm}{B-B_\pm} + \sum_{i \leq n+1} \frac{k_i}{(B-B_\theta)^i} \right]
\rlap{\ .}
\end{equation*}
Multiplying by $B-B_\pm$ and setting $B=B_\pm$ give
\begin{equation*}
k_\pm = \frac{1}{(B_\pm-B_\theta)^{n+1}}
\rlap{\ .}
\end{equation*}
The decomposition is completed by computing
\begin{equation*}
\frac{1}{B-B_\pm}\left(\frac{1}{(B-B_\theta)^{n+1}} - \frac{1}{(B_\pm-B_\theta)^{n+1}}\right) = - \sum_{k=0}^n \frac{1}{(B_\pm-B_\theta)^{n-k+1} (B-B_\theta)^{k+1}}
\rlap{\ .}
\end{equation*}
The pole decomposition finally reads
\begin{eqnarray*}
\pm \frac{F(B)}{B-B_\pm} &=& \pm \kappa \lambda_\infty
\\
&& \pm \left[ \kappa \lambda_\infty B_\pm + \gamma \left( \frac{( - B_\theta)^{n+1}}{(B_\pm-B_\theta)^{n+1}} -1 \right) \right]\frac{1}{B-B_\pm}
\\
&& \mp \gamma  \sum_{k=0}^n \frac{( - B_\theta)^{n+1}}{(B_\pm-B_\theta)^{n-k+1} (B-B_\theta)^{k+1}}
\end{eqnarray*}
which has to be plugged into equation \eqref{twoparts}.

We can now integrate \eqref{twoparts} analytically:
\begin{eqnarray}
A= A_T &+& \frac{1}{\frac{1}{2} \sigma^2 (B_+-B_-)} \times
\nonumber\\
\sum_\pm \Bigg(
&\pm& \left[ \kappa \lambda_\infty B_\pm + \gamma \left( \frac{1}{(1-\theta B_\pm)^{n+1}} -1 \right) \right] \ln\!\left( \frac{B-B_\pm}{B_T-B_\pm} \right)
\nonumber\\
&\mp& \gamma \frac{1}{(1-\theta B_\pm)^{n+1}} \ln\!\left( \frac{1-\theta B}{1-\theta B_T} \right)
\nonumber\\
&\pm& \gamma \sum_{k=1}^n \frac{1}{k} \frac{1}{(1-\theta B_\pm)^{n-k+1}}
\left[\frac{1}{(1-\theta B)^k} - \frac{1}{(1-\theta B_T)^k}\right]\Bigg)
\label{solA}
\rlap{\ .}
\end{eqnarray}
For $\beta \neq 0$ or $\eta \neq 0$, there are additional contributions to $A$
\begin{equation}
-\beta (1-e^x) (T-t) -\eta (1-e^y) (T-t)
\rlap{\ .}
\label{solAbis}
\end{equation}

Equations \eqref{solB}, \eqref{solA} and \eqref{solAbis} together with terminal conditions $B_T = w$ and $A_T = 0$ provide an analytical solution for the four-dimensional characteristic function of the joint distribution of $\widetilde{L}_T$, $\widetilde{N}_T$, $\lambda_T$, $Q_T$ and $R_T$.

The formulas computed correspond to the case where all parameters are constant between the computation time and the maturity $T$ for which the distributions must be computed. If parameters are piecewise constant, the integration has to be performed backwards, with intermediary boundary conditions $B_{T_i}$ and $A_{T_i}$.

\subsubsection{Loss distribution and moments}

An inverse Fourier transform gives the Loss distribution from the characteristic function
\begin{equation}
\mathbb{E}_t \left( e^{i p \widetilde{L}_T} \right) = e^{A(t,T;i p,0,0,0,0) + B(t,T;i p,0,0,0,0) \lambda_t + i p \widetilde{L}_t}
\label{caracL}
\end{equation}
where $u= i p $ and $v=w=y=0$.

The number of defaults can be obtained similarly by Fourier-inverting
\begin{equation*}
\mathbb{E}_t \left( e^{i p \widetilde{N}_T} \right) = e^{A(t,T;0,i p,0,0,0) + B(t,T;0,i p,0,0,0) \lambda_t + i p \widetilde{N}_t}
\end{equation*}
where $v=i p$ and $u=w=y=0$.
Computations involving $\widetilde{N}_t$ are similar to ones with $\widetilde{L}_t$ so in what follows we focus on $\widetilde{L}_t$.

We have to take into account the fact that the whole basket may default, which is modelled by the nonnegative stochastic variable $Q$. $Q>0$ means that all names have defaulted. We write
\begin{equation*}
\mathbbm{1}_{Q_T=0} = e^{-\infty Q_T}
\rlap{\ .}
\end{equation*}
This allows to compute quantities like
\begin{eqnarray*}
\mathbb{E}_t \left( e^{i p \widetilde{L}_T} \mathbbm{1}_{Q_T=0} \right) &=& \mathbb{E}_t \left( e^{i p \widetilde{L}_T -\infty Q_T} \right)
\\
&=& e^{A(t,T;i p,0,0,-\infty,0) + B(t,T;i p,0,0,-\infty,0) \lambda_t + i p \widetilde{L}_t}
 \mathbbm{1}_{Q_t=0}
\end{eqnarray*}
and  similarly
\begin{equation*}
\mathbb{E}_t \left( e^{i p \widetilde{N}_T} \mathbbm{1}_{Q_T=0} \right) = e^{A(t,T;0,i p,0,-\infty,0) + B(t,T;0,i p,0,-\infty,0) \lambda_t + i p \widetilde{N}_t}
 \mathbbm{1}_{Q_t=0}
\rlap{\ .}
\end{equation*}
We need in addition
\begin{equation*}
\mathbb{E}_t \left( \mathbbm{1}_{Q_T>0} \right) = 1-\mathbb{E}_t \left( \mathbbm{1}_{Q_T=0} \right) = 1-e^{A(t,T;0,0,0,-\infty,0) + B(t,T;0,0,0,-\infty,0) \lambda_t}
 \mathbbm{1}_{Q_t=0}
\rlap{\ .}
\end{equation*}
In fact, $x=-\infty$ means that $e^x$ is set to 0 in formulas of theorem~\ref{thm}.

The counterparty risk is handled similarly through variable $Rt$ which is null until the counterparty defaults. In the following, we do not take it into account in formulas in order not to overload formulas. It is straightforward to extend them to this case.

The loss mean $\mathbb{E}(\widetilde{L}_T)$ or the loss variance $\mathbb{E}\!\left(\left(\widetilde{L}_T\right)^2 - {\mathbb{E}\!\left(\widetilde{L}_T\right)}^2\right)$ can be computed in two ways. Either the Taylor expansion of equation~\eqref{caracL} in $p \sim 0$ is computed, which produces all moments up to a given order, or differential equations are directly written and solved for the moments.

\subsubsection{Time dependency}
\label{timeDependency}

To be able to match market prices of index CDS at all maturities, some parameters must be time-dependent. There are several possible choice. The long-run intensity $\lambda_\infty$ can be taken piecewise constant for example. An other possibility is a change of coordinate on the time variable. $t$ is taken to be an increasing function $t(\tau)$ of real time $\tau$. If time goes faster, defaults are more likely to happen in a given period of time, if times goes slower defaults have lower probability. For example if one wants to fit exactly the index CDS quotes at several maturity, the simplest solution is to take a piecewise affine function for $f$, with slope change at CDS maturities. When the slope changes by a factor $a$, this is equivalent to a scaling of parameters
\begin{equation*}
\begin{array}{cclc}
\lambda & \longrightarrow & a &\lambda \\
\gamma & \longrightarrow & a &\gamma \\
\theta & \longrightarrow & a &\theta \\
n & \longrightarrow & & n \\
\kappa & \longrightarrow & a &\kappa \\
\sigma & \longrightarrow & a &\sigma \\
\alpha & \longrightarrow & &\alpha \\
\beta & \longrightarrow & a &\beta \\
\xi & \longrightarrow & &\xi \\
\zeta & \longrightarrow & &\zeta \\
\eta & \longrightarrow & a &\eta
\rlap{\ .}
\end{array}
\end{equation*}

\subsection{Finite basket}

The model described really corresponds to an infinite pool of issuers: there can be any number of defaults, even larger than the basket considered. There is no exhaustion of the issuers pool. Let us consider now a basket of $N_M$ issuers. We denote by $N_t$ the number of default in the finite basket and $L_t$ the corresponding loss.

\subsubsection{Truncation}

A first solution would be to truncate the number of defaults to the basket size $N_M$
and the loss to $L_M = l_1 N_M$:
\begin{eqnarray*}
N_t &=& \min(\widetilde{N}_t,N_M)
\\
L_t &=& \min(\widetilde{L}_t,L_M)
\rlap{\ .}
\end{eqnarray*}
The expected loss $\mathbb{E}(L_t)$ may be approximated to $\mathbb{E}(\widetilde{L}_t)$ if the contributions from $\widetilde{L}_t > L_M$ are negligible. If this approximation is not valid, a lengthier computation must be performed: a CDO tranche with detachment point $L_M$ has to be priced to get the value of a index CDS. This would slow down our calibration process. Moreover, the probability that a given issuer defaults in a basket where many other issuers have already defaulted would be much greater than its default probability in a basket without defaults yet, for the same default intensity: the number of default per unit of time does not depend on the number of issuers alive in the basket. In some sense this is not a real problem: it introduces some default contagion. On the other hand, this default correlation is not controlled and does not come from the default intensity process.

\subsubsection{Subpool and uniform probability}

We therefore prefer an other solution. We consider that the finite pool is a part of the infinite pool. The infinite pool can be considered to model the whole economy. With uniform probability on pool components, the probability that a default of the infinite pool concerns the finite pool is proportional to its number of elements. Starting from a finite pool size of $N_M$, after $N$ defaults on the finite basket a given default of the infinite pool has a probability proportional to $N_M-N$ to occur in the basket. When the number of defaults in the basket is $N=N_M$ there is no more default: the number of defaults is correctly bounded. We normalize this probability as $\frac{N_M-N}{N_M}$. There is no loss of generality with this choice of normalization: the model would be the same if this probability is multiplied by some factor and the intensity process on the infinite pool is divided by the same factor.

We forget for the moment the probability of a simultaneous default of the whole basket. We can compute the probability of $k$ defaults in the credit basket as a sum over probabilities of $j$ defaults in the infinite pool weighted by the probability that $j$ defaults in the infinite pool gives $k$ defaults in the finite basket. Denoting by $p_{jk}$ this probability we have
\begin{equation}
\mathbb{P}(N_T = k \, \vert \,  Q_T = 0) = \sum_{j=0}^{\infty}  \mathbb{P}(\widetilde{N}_T = j \, \vert \,  Q_T = 0) p_{jk} \rlap{\ .}
\label{eq-change-pool}
\end{equation}
The $p_{kj}$ coefficients satisfy the recursion
\begin{equation}
p_{j+1,k} = \frac{k}{N_M} p_{j,k} + \frac{N_M-k+1}{N_M} p_{j,k-1}
\label{eq-pool-recur}
\end{equation}
with initial condition
\begin{equation}
p_{0,k} = \delta_{k0}  \rlap{\ .}
\label{eq-pool-ini}
\end{equation}

It is possible to find a closed form solution to this set of equations. The observation of the solution for $k=0,1,2$ leads to the hypothesis that $p_{jk}$, $k<j$, can be expanded in eigenmodes
\begin{equation*}
p_{jk} = \sum_{l=0}^{N_M} a_{kl} \left( \frac{l}{N_M} \right)^j \rlap{\ .}
\end{equation*}
Equation \eqref{eq-pool-recur} is satisfied if and only if
\begin{equation*}
(l-k) a_{kl} = (N-M-k+1) a_{k-1,l}
\end{equation*}
which leads to
\begin{equation*}
\begin{array}{rcll}
a_{kl} &=& \displaystyle  \binom{N_M-l}{k-l} a_{ll} (-1)^{k-l} & k \geq l
\\
a_{kl} &=& 0 & k<l \rlap{\ .}
\end{array}
\end{equation*}
The coefficients $a_{ll}$ are given by the initial condition \eqref{eq-pool-ini}. The expression for $p_{0k}$ is triangular and can be computed recursively.
To shorten the computation, we introduce a polynomial in some variable $Z$:
\begin{equation*}
P(Z) = \sum_{k=0}^{N_M} p_{0k} Z^k
\rlap{\ .}
\end{equation*}
The initial condition $p_{0k} = \delta_{k0}$ is equivalent to $P(Z) = 1$.
Using our solution for $p_{0k}$ in terms of $a_{kl}$ this polynomial is
\begin{equation*}
P(Z) = \sum_{k=0}^{N_M} \sum_{l=0}^{k} \binom{N_M-l}{k-l} a_{ll} (-1)^{k-l}   Z^k \rlap{\ .}
\end{equation*}
Switching sums in $k$ and $l$ we can sum over $k$ and get
\begin{equation*}
P(Z) = \sum_{l=0}^{N_M} a_{ll} (1-Z)^{N_M-l} Z^l = 1\rlap{\ .}
\end{equation*}
The unique solution to this equation is
\begin{equation*}
a_{ll} = \binom{N_M}{l}
\rlap{\ .}
\end{equation*}
After some rearrangement the solution finally is
\begin{equation}
p_{jk} = \binom{N_M}{k} (-1)^k \sum_{l=0}^{k} \binom{k}{l} (-1)^l \left(\frac{l}{N_M}\right)^j \rlap{\ .}
\label{eq-pool-sol}
\end{equation}

Using this solution and taking into account the disaster event modelled by $Q$, the probability of default in the finite basket \eqref{eq-change-pool} is
\begin{equation}
\mathbb{P}(N_T = k) = \binom{N_M}{k} (-1)^k \sum_{l=0}^{k} \binom{k}{l} (-1)^l \, \mathbb{E}\!\left[ \left(\frac{l}{N_M}\right)^{\widetilde{N}_T} \mathbbm{1}_{Q_T=0} \right] + \delta_{k,N_M}   \mathbb{E}\!\left(\mathbbm{1}_{Q_T>0} \right)
\end{equation}
where the expected values on the right-hand side are given by theorem~\ref{thm}:
\begin{eqnarray}
\mathbb{E}\!\left[ \left(\frac{l}{N_M}\right)^{\widetilde{N}_T}  \mathbbm{1}_{Q_T=0} \right] &=& e^{A\!\left(0,T,0,\ln\!\left( \frac{l}{N_M}\right),0,-\infty,0\right)+B\!\left(0,T,0,\ln\!\left( \frac{l}{N_M}\right),0,-\infty,0\right) \lambda_0}
\\
\mathbb{E}\!\left( \mathbbm{1}_{Q_T>0} \right) &=& 1-e^{A\!\left(0,T,0,0,0,-\infty,0\right)+B\!\left(0,T,0,0,0,-\infty,0\right) \lambda_0}
 \rlap{\ .}
\end{eqnarray}
More generally, the conditional probability at time $t$ is given by a slightly modified formula:
\begin{multline}
\mathbb{P}_t(N_T = k) = \mathbbm{1}_{k \geq N_t} \binom{N_M-N_t}{k-N_t} (-1)^k \sum_{l=N_t}^{k} \binom{k-N_t}{l-N_t} (-1)^{l} \, \mathbb{E}_t\!\left[ \left(\frac{l}{N_M}\right)^{\widetilde{N}_T-\widetilde{N}_t}  \mathbbm{1}_{Q_T=0} \right] 
\\
+ \delta_{k,N_M}  \mathbb{E}_t\!\left(\mathbbm{1}_{Q_T>0} \right)
\end{multline}
where the expected value on the right-hand side are given by theorem~\ref{thm}:
\begin{eqnarray}
\mathbb{E}_t\!\left[ \left(\frac{l}{N_M}\right)^{\widetilde{N}_T-\widetilde{N}_t}  \mathbbm{1}_{Q_T=0}\right] &=& e^{A\!\left(t,T,0,\ln\!\left( \frac{l}{N_M}\right),0,-\infty,0\right)+B\!\left(t,T,0,\ln\!\left( \frac{l}{N_M}\right),0,-\infty,0\right) \lambda_t} \mathbbm{1}_{Q_t=0}
\\
\mathbb{E}_t\!\left( \mathbbm{1}_{Q_T>0} \right) &=& 1-e^{A\!\left(t,T,0,0,0,-\infty,0\right)+B\!\left(t,T,0,0,0,-\infty,0\right) \lambda_0} \mathbbm{1}_{Q_t=0}
 \rlap{\ .}
\end{eqnarray}

We have finally obtained a closed formula for the default distribution as a finite sum. For small enough $N_M$ or $k$ it can be used. However for large values of $N_M$ and $k$ this formula is completely useless for numerical applications: it is a alternating sum of numbers with very large coefficients. For example for $N_M=125$ and $k=62$, which corresponds to the highest contribution to the standard senior tranche on CDX with a recovery rate of 40\%, the highest coefficient has an order of magnitude of $10^{54}$ and the final result is between 0 and 1: the computations should be done with very high accuracy and cannot be done in a direct way on current processors.
In the remaining parts of this section we therefore give other computation techniques more suitable to numerical implementation.

\subsubsection{Expected loss}
\label{sec-loss}

In order to compute CDS prices, the expected loss
\begin{equation}
\mathbb{E}(L_T ) = l_1 \mathbb{E}(N_T )
\end{equation}
can be computed from the mean jump size $l_1$, equal to 1 minus the average recovery rate, and the expected number of defaults. To compute the latter quantity, we introduce the conditional mean value
\begin{equation*}
e_j = \mathbb{E}(N_T | \widetilde{N}_T=j , Q_T = 0) =  \sum_{k=0}^{N_M} p_{jk} k
\end{equation*}
which does not depend on $T$ as the probability of jumps in the basket does not depend explicitly on time. In a first step we condition all probabilities by $Q_T=0$, \emph{i.e.} we do not consider a default of the whole infinite pool.
The recursion equation \eqref{eq-pool-recur} translates for the mean value to
\begin{equation*}
e_{j+1} - e_j = \sum_{k=0}^{N_M} p_{jk} \frac{N_M-k}{N_M} =  \frac{N_M-e_j}{N_M}
\rlap{\ .}
\end{equation*}
From the initial value $e_0 = 0$ the solution is
\begin{equation*}
e_j = N_M \left[ 1- \left( 1-\frac{1}{N_M}\right)^j \right] \rlap{\ .}
\end{equation*}
The unconditional expected number of defaults is obtained by averaging over the number of defaults $j$ in the infinite pool and taking into account the default of the whole infinite pool:
\begin{equation*}
\mathbb{E}(N_T) = N_M \left( 1- \mathbb{E}\!\left[ \left( 1-\frac{1}{N_M}\right)^{\widetilde{N}_T} \mathbbm{1}_{Q_T=0} \right] \right) \rlap{\ .}
\end{equation*}
The expected value on the right-hand side corresponds to the characteristic function we have computed in theorem~\ref{thm}, with $v=\ln\!\left(1-\frac{1}{N_M}\right)$, $x=-\infty$, $u=w=y=0$:
\begin{equation}
\mathbb{E}(N_T) = N_M \left[ 1-e^{A\!\left(0,T;0,\ln\!\left(1-\frac{1}{N_M}\right),0,-\infty,0\right) + B\!\left(0,T;0,\ln\!\left(1-\frac{1}{N_M}\right),0,-\infty,0\right) \lambda_0} \right] \rlap{\ .}
\end{equation}

We find similarly the conditional average number of default:
\begin{equation*}
\mathbb{E}_t(N_T) = N_M - (N_M-N_t) \mathbb{E}_t\!\left[ \left( 1-\frac{1}{N_M}\right)^{\widetilde{N}_T-\widetilde{N}_t} \mathbbm{1}_{Q_T=0} \right] 
\rlap{\ .}
\end{equation*}
Theorem~\ref{thm} gives the explicit expression
\begin{equation}
\mathbb{E}_t(N_T) =  N_M - (N_M-N_t) e^{A\!\left(t,T;0,\ln\!\left(1-\frac{1}{N_M}\right),0,-\infty,0\right) + B\!\left(t,T;0,\ln\!\left(1-\frac{1}{N_M}\right),0,-\infty,0\right) \lambda_t}  \mathbbm{1}_{Q_t=0} \rlap{\ .}
\label{eqEL}
\end{equation}

\subsubsection{Expected tranche loss}

For CDOs we have to compute the expected value of the tranche loss. We consider first the case of fixed recovery rate, with default size $l_1$ for all components.
From a tranche with attachment point $a$ and detachment point $d$, the tranche loss is
\begin{equation*}
L_T^{a,d} = (d-a) + (a-L_T)^+ - (d-L_T)^+
\rlap{\ .}
\end{equation*}

The expected value of
\begin{equation*}
\mathbb{E}\!\left[ (K - L_T)^+ \right] = \sum_{k=0}^{N_M} \mathbb{P}(N_T=k) (K-l_1 k)^+
\end{equation*}
can be rewritten using \eqref{eq-change-pool} as
\begin{equation*}
\mathbb{E}\!\left[ (K - L_T)^+ \right] = \sum_{j=0}^{\infty} \left[ \sum_{k=0}^{N_M} p_{jk} (K-l_1 k)^+ \right] \mathbb{P}\!\left(\widetilde{N}_T=j , Q_T=0\right) 
\rlap{\ .}
\end{equation*}
(We suppose that the tranche does not contain the end of the basket: $K<l_1 N_M$.)
This quantity is the integral of the function
\begin{equation*}
f_K(j) =  \sum_{k=0}^{N_M} p_{jk} (K-l_1 k)^+
\end{equation*}
against the probability distribution of $\widetilde{N}_T$.

Theorem~\ref{thm} gives the Fourier transform of this distribution. We can invert it at each time $T$ for which we need to compute the tranche loss and integrate against $f_K$. When calibrating, these Fourier transforms has also to be done for all sets of model parameters. Alternatively, we can compute the Fourier transform of $f_K$, which has to be done only once for each value of $K$ and integrate it against the characteristic function of $\widetilde{N}_T$. This considerably reduces the number of numerical Fourier transform to perform.

For each value of $K$ needed, the Fourier transform of $f_K$
\begin{equation*}
\hat{f}_K(p) = \sum_{j=0}^{\infty} e^{ipj} f_K(j)  = \sum_{j=0}^{\infty}  e^{ipj} \sum_{k=0}^{N_M} p_{jk} (K-l_1 k)^+
\end{equation*}
is computed numerically, using FFT for example. A bound on $j$ has to be estimated. Matrix elements $p_{jk}$ have an explicit form given in equation~\eqref{eq-pool-sol} but we have seen this expression is often useless numerically. We therefore compute these numbers directly using the recursion~\eqref{eq-pool-recur} from initial condition \eqref{eq-pool-ini}. As these numbers does not depend on the model parameters or even the attachment or detachment point but only on the basket size $N_M$, this computation has to be done only once for all instruments on a basket of size $N_M$. The Fourier transform itself has to be done only once for each detachment point $K$, basket size $N_M$ and recovery rate $1-l_1$.

We finally obtain $\mathbb{E}\!\left[ (K - L_T)^+ \right]$ as
\begin{equation*}
\mathbb{E}\!\left[ (K - L_T)^+ \right] = \frac{1}{2\pi} \int_{-\infty}^{+\infty} \hat{f}_K(-p) \mathbb{E}\!\left(e^{ip \widetilde{N}_T} \mathbbm{1}_{Q_T=0}\right)
\end{equation*}
or more explicitly using theorem~\ref{thm}
\begin{equation}
\mathbb{E}\!\left[ (K - L_T)^+ \right] = \frac{1}{2\pi} \int_{-\infty}^{+\infty} \hat{f}_K(-p)
e^{ A( 0,T;0,i p,0,-\infty,0) + B( 0,T;0,i p,0,-\infty,0)  \lambda_0}
\rlap{\ .}
\label{eqPut}
\end{equation}

For N\textsuperscript{th}-to-default, we need to compute at every time step the quantity $\mathbb{E}\!\left( \mathbbm{1}_{N_T < k} \right) = \mathbb{P}(N_T<k)$. It is computed similarly for $k<N_M$ as
\begin{eqnarray}
\mathbb{E}\!\left( \mathbbm{1}_{N_T < k} \right) &=& \frac{1}{2\pi} \int_{-\infty}^{+\infty} \hat{g}_k(-p)  \mathbb{E}\!\left(e^{ip \widetilde{N}_T}  \mathbbm{1}_{Q_T=0}\right)
\nonumber
\\
& =&  \frac{1}{2\pi} \int_{-\infty}^{+\infty} \hat{g}_k(-p)
e^{ A( 0,T;0,i p,0,-\infty,0) + B( 0,T;0,i p,0,-\infty,0)  \lambda_0}
\label{eqPutDigital}
\end{eqnarray}
with
\begin{equation*}
\hat{g}_k(p) = \sum_{j=0}^{\infty} e^{ipj} \mathbbm{1}_{N_T < k}  = \sum_{j=0}^{\infty}  e^{ipj} \sum_{k'=0}^{k-1} p_{jk'}
\rlap{\ .}
\end{equation*}

\subsubsection{Stochastic recovery}
\label{sec-stoch-reco}

Stochastic recovery rate can be included in the framework. The loss process is no longer $L_T = l_1 N_T$. If recovery rates are independent of each other, $L_T$ is given by the convolution of the jump distributions up to the ${N_T}^{\textrm{th}}$ jump.

When computing tranche losses, the function $f_K(j)$ has to be modified. More precisely, $(K-l_1 k)^+$ has to be replaced by the expected value of $(K-L_T)^+$ conditional to $N_T=k$: it is the integral of this function weighted by the distribution of $L_T$ conditional to $N_T=k$, \emph{i.e.} the convolution of distributions of the $k$ first jumps.

If all recovery rates have the same mean $1-l_1$, the expected loss is still given by $E(L_T) = l_1 E(N_T)$. If the recovery rate mean depends on the number of past jumps, it is no longer true. In many cases it is still possible to compute a simple formula for the expected loss. In other cases, one can always use the algorithm for tranche losses for a tranche covering the whole basket.

\section{Pricing}
\label{sec-pricing}

\subsection{Index CDS}
\label{sec-cds}

The price of a CDS is given by the difference between a protection leg which pays the amount of defaults on the basket\footnote{We write payoffs with integrals to simplify notations; they are in fact discrete sums.}
\begin{equation*}
CL_{CDS} = \int_0^T d L_{\tau} ZC(0,\tau)
\end{equation*}
and a premium leg which pays a fixed premium $s$
\begin{equation*}
FL_{CDS} = s \int_0^T \ud \tau (N_M - N_{\tau}) ZC(0,\tau)
\rlap{\ .}
\end{equation*}
$T$ is the maturity of the contract, time 0 is the pricing date and $ZC(0,T)$ is the zero coupon discount factor. The price at time 0 is the risk-neutral expected value
\begin{equation}
\mathbb{E}( CL_{CDS}-FL_{CDS}) = \int_0^T d \mathbb{E}(L_{\tau}) ZC(0,\tau)
- s \int_0^T \ud \tau (N_M - \mathbb{E}(N_{\tau})) ZC(0,\tau)
\label{CDSprice}
\rlap{\ .}
\end{equation}
The expected values which appear in the right-hand side can be computed using equation \eqref{eqEL}. With a recovery rate of fixed mean size $1-l_1$, the expected value of $L_\tau$ is $\mathbb{E}(L_\tau) = l_1 \mathbb{E}(N_\tau)$. For more complex cases, see section~\ref{sec-stoch-reco}.

\subsection{CDO tranche}
\label{sec-cdo}

For a CDO tranche of attachment point $a$ and detachment point $d$, the tranche loss is
\begin{equation*}
L_\tau^{a,d} = 
\begin{cases}
0 &\text{if } L_\tau \leq a
\\
L_\tau - a &\text{if } a \leq L_\tau \leq d
\\
(d - a) &\text{if } L_\tau \geq d
\rlap{\ .}
\end{cases}
\end{equation*}
The credit and premium legs are
\begin{eqnarray*}
CL_{CDO}^{a,d} &=& \int_0^T d L_{\tau}^{a,d} ZC(0,\tau)
\\
FL_{CDO}^{a,d} &=& s \int_0^T \ud \tau \left((d-a) - L_{\tau}^{a,d}\right) ZC(0,\tau)
\rlap{\ .}
\end{eqnarray*}

Introducing
\begin{equation*}
M_\tau^{a,d} = (d-a) - L_\tau^{a,d}
\end{equation*}
The price is the risk-neutral expected value
\begin{equation}
\mathbb{E}( CL_{CDO}^{a,d}-FL_{CDO}^{a,d}) = - \int_0^T d \mathbb{E}\!\left(M_{\tau}^{a,d}\right) ZC(0,\tau)
- s \int_0^T \ud \tau \, \mathbb{E}\!\left(M_{\tau}^{a,d}\right) ZC(0,\tau)
\label{CDOprice}
\rlap{\ .}
\end{equation}
We decompose the remaining notional on the tranche as
\begin{equation*}
M_\tau^{a,d} =(d-L_\tau)^+ - (a-L_\tau)^+
\rlap{\ .}
\end{equation*}
Its expected value is
\begin{equation*}
\mathbb{E}\!\left(M_\tau^{a,d}\right) = \mathbb{E}\!\left[ (d-L_\tau)^+ \right] - \mathbb{E}\!\left[ (a-L_\tau)^+ \right]
\end{equation*}
which can be computed using equation \eqref{eqPut}. This formula involves an integral in Fourier space. In fact only one integral in Fourier space is needed when computing the price of a CDO tranche: by linearity, all quantities in the Fourier space can be summed before doing a final integral.

\subsection{N\textsuperscript{th}-to-default}

For a $k^{\text{th}}$-to-default let us introduce
\begin{equation*}
I_\tau^{k} = \mathbbm{1}_{N_\tau < k}
\rlap{\ .}
\end{equation*}
The credit and premium legs are
\begin{eqnarray*}
CL_{k} &=& - (1-R) \int_0^T d I_{\tau}^{k} ZC(0,\tau)
\\
FL_{k} &=& s \int_0^T \ud \tau \left(1 - I_{\tau}^{k}\right) ZC(0,\tau) \rlap{\ .}
\end{eqnarray*}

The price is the risk-neutral expected value
\begin{equation}
\mathbb{E}( CL_k-FL_k) = -l_1 \int_0^T d \mathbb{E}\!\left(I_{\tau}^{k}\right) ZC(0,\tau)
- s \int_0^T \ud \tau \left( 1 - \mathbb{E}\!\left(I_{\tau}^{k}\right) \right) ZC(0,\tau)
\rlap{\ .}
\label{Nthprice}
\end{equation}
The probability
\begin{equation*}
\mathbb{E}\!\left(I_{\tau}^{k}\right) = \mathbb{P}(N_\tau<k)
\end{equation*}
is computed using equation \eqref{eqPutDigital}.

\subsection{Options on Index CDS}

As the model we consider is a dynamic model, options on index CDS (swaptions) can also be priced. The owner of such an option get the right but not the obligation at some maturity to enter into an index CDS at a rate $K$ fixed in advance. A \emph{Payer Swaption} gives the right to pay the premium leg and get protection on the basket. A \emph{Receiver Swaption} gives the right to receive the premium leg and sell protection on the basket.

The price of the swap at exercise date $t$ is a stochastic process which depends on $\lambda_t$, $L_t$ and/or $N_t$. Expected values in formula \eqref{CDSprice} have to be evaluated at date $t$:
\begin{equation*}
\mathbb{E}_t( CL_{CDS}-FL_{CDS}) = \int_t^T d \mathbb{E}_t(L_{\tau}) ZC(t,\tau)
- K \int_t^T \ud \tau \left[ N_M - \mathbb{E}_t(N_{\tau}) \right] ZC(t,\tau)
\rlap{\ .}
\end{equation*}
A \emph{Front End Protection} clause can be included in a Payer Swaption: in this case if there are some defaults before the option expiry the corresponding amount $L_t$ is paid at the exercise. A Receiver Swaption would not be exercised if defaults occur prior to expiry.
For a Payer swaption with Front-End Protection the price at exercise date $t$ is
\begin{equation*}
\left[ \mathbb{E}_t( CL_{CDS}-FL_{CDS}) + L_t \right]^+
\end{equation*}
whereas for a Receiver swaption it is
\begin{equation*}
\left[ \mathbb{E}_t( FL_{CDS} - CL_{CDS} ) - L_t\right]^+
\rlap{\ .}
\end{equation*}
The price of the Payer swaption is the risk-neutral expected value
\begin{equation*}
SP = \mathbb{E}_0\!\left( \left[ \mathbb{E}_t\!\left( CL_{CDS} - FL_{CDS} \right) +L_t \right]^+ \right) ZC(0,t)
\end{equation*}
and similarly for the Receiver
\begin{equation*}
SR = \mathbb{E}_0\!\left( \left[ \mathbb{E}_t\!\left( FL_{CDS} - CL_{CDS} \right) -L_t \right]^+ \right) ZC(0,t)
\rlap{\ .}
\end{equation*}

In all cases the price is conditional to the values of $\lambda_t$, $L_t$ and/or $N_t$. The information available at exercise date is supposed to be fully contained in these variables.
The probability density of these variables is obtained in the following way. The joint probability distribution of $\widetilde{N}_t$ and $\lambda_t$ (conditional to $Q_t=0$) is computed by Fourier inversion from the characteristic function
\begin{equation*}
\mathbb{E}_0 \left( e^{v \widetilde{N}_t + w \lambda_t} \mathbbm{1}_{Q_t=0}\right)
\end{equation*}
computed using closed formulas \eqref{affine}, \eqref{solB} and \eqref{solA}.
This joint distribution is combined with the probabilities of $N_t$ and $L_t$ conditional to $\widetilde{N}_t$  to get the joint probability of $\lambda_t$, $N_t$ and $L_t$. 
With fixed default size $l_1$, the joint density is
\begin{multline*}
\ud\mathbb{P}(\lambda_t, N_t=k, L_t) = \sum_{j=0}^{\infty} \ud \mathbb{P}(\lambda_t, \widetilde{N}_t=j, Q_t=0) p_{jk} \delta(L_t-l_1 k) \ud L_t 
\\
+ \delta_{k,N_M} \ud \mathbb{P}(\lambda_t, Q_t>0) \delta(L_t-l_1 k) \ud L_t 
\ .
\end{multline*}

The conditional price is then integrated against this probability density. In the most general case, a 2-dimensional Fourier transform and a 3-dimensional integral should be performed. In fact $L_t$ and $N_t$ carry almost the same information and $L_t$ can be removed from the integration variables and supposed to be equal to $l_1 N_t$. If the recovery rate is deterministic this is exact, in other cases it is a slight approximation. In any case, we have to compute a 2-dimensional Fourier transform and a 2-dimensional integral. Making few approximations, the numerical computations can be reduced to few simple integrals, which can be useful in particular for calibration purpose. This will be detailed in section~\ref{sec-fast-swaption}.

\subsection{Options on CDO tranches and exotic products}

Options on CDO tranches can be priced exactly: the conditional price at the exercise date $t$ is integrated against the probability density of $\lambda_t$ and $L_t$. Other exotic payoffs can also be priced as the model allows to compute the expected value at any date of any future loss.

\subsection{Counterparty risk}

The default risk of a counterparty can be handled within our framework through the auxiliary process $R$ introduced in section~\ref{sec:model}. $R$ is set to 0 at the pricing date. It jumps of one unit when the counterparty defaults.

If a quantity will not be paid if the counterparty has defaulted, the characteristic functions of theorem~\ref{thm} are computed with the last parameter set to $y=-\infty$ (or equivalently $e^y=0$ in all formulas). In this limit, the contributions to the expected values are set to zero when $Q$ is positive.

\subsection{Large pool approximation}

\subsubsection{Limit}

The distribution of $N_T$ conditional to $\widetilde{N}_T$ and $Q_T=0$ that we have denoted by $p_{jk} = \mathbb{P}(N_T = k \vert \widetilde{N}_T=j , Q_T=0)$ has significant values around its mean value
\begin{equation*}
e_j = \mathbb{E}(N_T  \vert \widetilde{N}_T=j , Q_T=0) = N_M\left[1-\left(1-\frac{1}{N_M}\right)^j \right]
\rlap{\ .}
\end{equation*}
The variance can also be computed: it is
\begin{equation*}
v_j = \sum_{k=0}^{N_M} p_{jk} k^2 - e_j^2 = N_M \left[ \left( 1-\frac{1}{N_M} \right)^j + (N_M-1)\left( 1-\frac{2}{N_M} \right)^j -N_M \left( 1-\frac{1}{N_M} \right)^{2j} \right]
\rlap{\ .}
\end{equation*}
We consider a fraction $p$ of the portfolio, \emph{i.e.} an average number of defaults $e_j \sim p N_M$. It corresponds to
\begin{equation*}
j \sim \frac{\ln(1-p)}{\ln(1-1/N_M)}
\rlap{\ .}
\end{equation*}
At fixed portfolio percentage $p$ the variance asymptotic when $N_M$ goes to infinity is
\begin{equation*}
v_j \sim p(1-p) N_M
\rlap{\ .}
\end{equation*}
In proportion of the basket size $N_M$, the standard deviation is therefore
\begin{equation*}
\frac{\sqrt{v_j}}{N_M} \sim \sqrt{\frac{p(1-p)}{N_M}}
\rlap{\ .}
\end{equation*}
When $N_M$ goes to infinity, the distribution becomes peaked around its mean value.
Our large pool approximation consists in neglecting the dispersion of $N_T$ conditional to $\widetilde{N}_T$: we take
\begin{equation}
N_T \simeq N_M\left[1-\left(1-\frac{1}{N_M}\right)^{\widetilde{N}_T} \mathbbm{1}_{Q_T=0} \right]
\rlap{\ .}
\end{equation}

Similarly for the loss process we take
\begin{equation}
L_T = L_M\left[1-e^{-\mu \widetilde{L}_T} \mathbbm{1}_{Q_T=0}\right]
\end{equation}
with $L_M \simeq l_1 N_M$ and $\mu$ such that
\begin{equation*}
\phi_J(-\mu) = \int e^{-_mu l} \ud\nu_J(l) = 1-\frac{1}{N_M}
\rlap{\ .}
\end{equation*}
The value of the constants are chosen so that the mean of $L_T$ conditional to $\widetilde{N}_T$ is $l_1 N_T$.
(By definition the mean default size $l_1$ is $ l_1 = \int j \ud\nu_J(j)$.)

This is an acceptable approximation if the basket size $N_M$ is large enough. Indeed for low losses, $L \simeq \widetilde{L}$ and jump sizes are not reduced. With many defaults, the recovery rate distribution smooths the distribution so that the convolution implied by this approximation does not modify it too much. Moreover senior CDO tranches are usually wide which means that the short scale information of the distribution that we lose is not very important. The only case where it would be important to perform an exact computation would be the pricing of a N\textsuperscript{th}-to-default where the combinatorics has to be performed or at least approximated. As far as only the integral of the loss or default number on some range is considered, this approximation is reasonable.

\subsubsection{Expected loss}

In order to compute CDS prices, the expected loss
\begin{equation*}
\mathbb{E}(L_T ) = L_M \left[ 1-\mathbb{E}\!\left(e^{-\mu \widetilde{L}_T} \mathbbm{1}_{Q_T=0} \right) \right]
\end{equation*}
can be computed analytically: it corresponds to the characteristic function we have computed with $u=-\mu$, $x=-\infty$, $v=w=y=0$:
\begin{equation*}
\mathbb{E}(L_T) = L_M \left[ 1-e^{A\!\left(0,T;-\mu,0,0,-\infty,0\right) + B\!\left(0,T;-\mu,0,0,-\infty,0\right) \lambda_0} \right]
\rlap{\ .}
\end{equation*}

As we said above, constants $L_M$ and $\mu$ are such that it is equal to
\begin{equation*}
\begin{split}
\mathbb{E}(L_T )  &= l_1 \mathbb{E}(N_T) =   l_1 N_M \left( 1-\mathbb{E}\!\left[\left(1-\frac{1}{N_M} \right)^{\widetilde{N}_T} \mathbbm{1}_{Q_T=0} \right] \right)
\\
&= l_1 N_M \left[ 1-e^{A\!\left(0,T;0,\ln\!\left( 1-\frac{1}{N_M}\right),0,-\infty,0\right) + B\!\left(0,T;0,\ln\!\left( 1-\frac{1}{N_M}\right),0,-\infty,0\right) \lambda_0} \right]
\rlap{\ .}
\end{split}
\end{equation*}
By construction, this is in fact the exact expression of section \ref{sec-loss}.

\subsubsection{Expected tranche loss}

For CDOs, the pricing of section \ref{sec-cdo} using exact formulas is fast enough. However we show here how it can be done in our large pool approximation. It is similar to the pricing of equity options with Levy models.
We have to compute the expected value of the tranche loss. It can be obtained as the difference of two terms of the form
\begin{equation*}
\mathbb{E}\!\left[ (K - L_T)^+ \right]
= L_M \mathbb{E}\!\left[ \left( e^{-\mu \widetilde{L}_T} - e^{-\mu \widetilde{K}} \right)^+  \mathbbm{1}_{Q_T=0} \right]
\end{equation*}
with
\begin{equation}
\widetilde{K} = - \frac{1}{\mu} \ln\!\left( 1-\frac{K}{L_M} \right)
\label{eqK}
\rlap{\ .}
\end{equation}

To compute this quantity we introduce a second distribution $\ud \widehat{\mathbb{P}}(\widetilde{L}_T,Q_T)$ for the joint law of $\widetilde{L}_T$ and $Q_T$ with an explicit expression for the density function, such that it has the same mean value of $e^{-\mu\widetilde{L}_T} \mathbbm{1}_{Q_T=0}$ as the real distribution $\ud \mathbb{P}(\widetilde{L}_T,Q_T)$ coming from the double stochastic process, \emph{i.e.} the same mean value of $L_t$.

The difference between both distributions for this quantity is computed by inverse Fourier transform from
\begin{equation*}
\begin{split}
\int_{-\infty}^{+\infty} d \widetilde{K} e^{u \widetilde{K}} &\left(\mathbb{E} - \widehat{\mathbb{E}}\right)\!\left[ ( K - L_t )^+  \right]
\\
&= \int_{-\infty}^{+\infty} d \widetilde{K} e^{u \widetilde{K}} \int_{-\infty}^{\widetilde{K}} \left(\ud \mathbb{P}(\widetilde{L}_T,Q_T) - \ud \widehat{\mathbb{P}}(\widetilde{L}_T,Q_T) \right) L_M \left( e^{-\mu \widetilde{L}_T} - e^{-\mu \widetilde{K}} \right) \mathbbm{1}_{Q_T=0}
\\
&= L_M \int_{-\infty}^{+\infty} \left( \ud \mathbb{P}(\widetilde{L}_T,Q_T) - \ud \widehat{\mathbb{P}}(\widetilde{L}_T,Q_T) \right) \int_{\widetilde{L}_T}^{+\infty} \ud \widetilde{K} e^{u\widetilde{K}} \left( e^{-\mu \widetilde{L}_T} - e^{-\mu \widetilde{K}}  \right) \mathbbm{1}_{Q_T=0}
\\
&= L_M \int_{-\infty}^{+\infty} \left(\ud \mathbb{P}(\widetilde{L}_T,Q_T) - \ud \widehat{\mathbb{P}}(\widetilde{L}_T,Q_T) \right) \left( \frac{1}{u-\mu} - \frac{1}{u} \right) e^{\left( u - \mu \right) \widetilde{L}_T} \mathbbm{1}_{Q_T=0}
\\
&= \frac{\mu L_M}{u \left(u-\mu\right)} \left[ \mathbb{E}\!\left( e^{\left( u - \mu \right) \widetilde{L}_T} \mathbbm{1}_{Q_T=0} \right) - \widehat{\mathbb{E}}\!\left( e^{\left( u - \mu \right) \widetilde{L}_T} \mathbbm{1}_{Q_T=0} \right) \right]
\end{split}
\end{equation*}
where we have used the fact that by construction of $\ud \hat{\nu}_{\widetilde{L}_t}$
\begin{eqnarray*}
\int_{-\infty}^{+\infty} \left(\ud \mathbb{P}(\widetilde{L}_T,Q_T) - \ud \widehat{\mathbb{P}}(\widetilde{L}_T,Q_T) \right) 1 &=& 0
\\
\int_{-\infty}^{+\infty} \left(\ud \mathbb{P}(\widetilde{L}_T,Q_T) - \ud \widehat{\mathbb{P}}(\widetilde{L}_T,Q_T) \right) e^{-\mu \widetilde{L}_T} \mathbbm{1}_{Q_T=0} &=& 0
\rlap{\ .}
\end{eqnarray*}

Inverting this Fourier transform we get
\begin{equation*}
\mathbb{E}\!\left[ (K-L_T)^+ \right]
= \widehat{\mathbb{E}}\!\left[ (K-L_T)^+ \right]
+ \frac{1}{2\pi}\int_{-\infty}^{+\infty} dp \, e^{-i p \widetilde{K}} \frac{\mu L_M}{i p \left(i p-\mu\right)} \left[ \mathbb{E}\!\left( e^{\left( i p - \mu \right) \widetilde{L}_T} \right) - \widehat{\mathbb{E}}\!\left( e^{\left( i p - \mu \right) \widetilde{L}_T} \right) \right]
\end{equation*}
where $\widetilde{K}$ is given by equation \eqref{eqK}. Expected values $\widehat{\mathbb{E}}$ under $\ud \widehat{\mathbb{P}}(\widetilde{L}_T,Q_T)$ are computed in an explicit way. $\mathbb{E}\!\left( e^{\left( i p - \mu \right) \widetilde{L}_T} \mathbbm{1}_{Q_T=0} \right)$ is given by theorem \ref{thm}:
\begin{equation*}
\mathbb{E}\!\left( e^{\left( i p - \mu \right) \widetilde{L}_T} \mathbbm{1}_{Q_T=0} \right) =
e^{ A\!\left( 0,T;i p - \mu,0,0,-\infty,0\right) + B\!\left( 0,T;i p - \mu,0,0,-\infty,0\right) \lambda_0}
\rlap{\ .}
\end{equation*}

There are many choices for the auxiliary distribution $\ud \widehat{\mathbb{P}}(\widetilde{L}_T,Q_T)$. We give two examples:
\begin{itemize}
\item The simplest choice is
\begin{equation*}
\ud \widehat{\mathbb{P}}(\widetilde{L}_T,Q_T) = \delta\!\left(\widetilde{L}_T+\frac{1}{\mu} \ln\!\left[ \mathbb{E}\!\left( e^{-\mu \widetilde{L}_t } \mathbbm{1}_{Q_T=0} \right) \right] \right) \delta_{Q_T,0}
\rlap{\ .}
\end{equation*}
It corresponds to the computation of the "time value" of the tranche loss by Fourier transform, which is then added to the "intrinsic value"
\begin{equation*}
\widehat{\mathbb{E}}\!\left[ (K-L_T)^+ \right] = L_M \left( \mathbb{E}\!\left[ e^{-\mu \widetilde{L}_T} \mathbbm{1}_{Q_T=0} \right] - e^{-\mu \widetilde{K}} \right)^+
\rlap{\ .}
\end{equation*}
\item Numerically it may be a better idea to take the product of a distribution on $Q_T$ such that
\begin{equation*}
\widehat{\mathbb{P}}(Q_T=0) = \mathbb{P}(Q_T=0)
\end{equation*}
and a Poisson distribution on $\frac{\widetilde{L}_T}{l_1}$ with parameter
\begin{equation*}
\hat{\lambda}_T = \frac{-\ln\!\left[ \frac {\mathbb{E}\!\left( e^{-\mu \widetilde{L}_T } \mathbbm{1}_{Q_T=0} \right)}{\mathbb{E}\!\left( \mathbbm{1}_{Q_T=0} \right)} \right]}{1-e^{-\mu l_1}}
\rlap{\ .}
\end{equation*}
For this distribution the quantity $\widehat{\mathbb{E}}\!\left[ (K-L_T)^+ \right] = L_M \widehat{\mathbb{E}}\!\left[ \left( e^{-\mu \widetilde{L}_T} - e^{-\mu \widetilde{K}} \right)^+ \mathbbm{1}_{Q_T=0} \right]$ can be computed exactly from the Poisson distribution as a finite sum.
\end{itemize}

Numerically, one can perform the Fourier inversion either using the FFT algorithm or using a one-dimensional quadrature. If many detachment points $K$ are needed the FFT is the natural choice. In the usual case where only one or a few tranches are priced (for example the five tranches of iTraxx) performing the one-dimensional integrals directly will be faster.

\subsubsection{Fast pricing of index swaptions}
\label{sec-fast-swaption}

\begin{proposition}
The price of a Receiver Swaption of strike $K$, maturity $t$ and tenor $T-t$ with Front End Protection can be computed (approximately) as a one-dimensional integral in Fourier space.

Let define the function
\begin{equation*}
h(\tau) = \left[ \delta(\tau-T) + \frac{K}{l_1} + r(\tau) \right] ZC(t,\tau)
\end{equation*}
where $r$ is the instantaneous interest rate
and
\begin{equation*}
L_\star(\lambda_t) = \frac{1}{\mu} \ln\!\left( \int_t^T \ud \tau \, h(\tau) e^{A\!\left(t,\tau;-\mu,0,0,-\infty,0\right) + B\!\left(t,\tau;-\mu,0,0,-\infty,0\right) \lambda_t} \right)
\rlap{\ .}
\end{equation*}
If $L_\star(0)$ is negative, the Receiver Swaption is never exercised and its price is zero.
Otherwise, let $\Lambda \geq 0$ be a typical value\footnote{Which may be as simple as $\Lambda=0$.} for $\lambda_t$ such that $L_\star(\Lambda) \geq 0$. Introducing
\begin{eqnarray*}
\alpha_\Lambda &=& \frac{1}{\mu} \ln\!\left( \int_t^T \ud \tau \, h(\tau) e^{A\!\left(t,\tau;-\mu,0,0,-\infty,0\right) + B\!\left(t,\tau;-\mu,0,0,-\infty,0\right) \Lambda}\right)
\\
\beta_\Lambda &=& \frac{1}{\mu} \frac{\displaystyle \int_t^T \ud \tau h(\tau) e^{A\!\left(t,\tau;-\mu,0,0,-\infty,0\right) + B\!\left(t,\tau;-\mu,0,0,-\infty,0\right) \Lambda} B\!\left(t,\tau;-\mu,0,0,-\infty,0\right)}{\displaystyle \int_t^T \ud \tau h(\tau) e^{A\!\left(t,\tau;-\mu,0,0,-\infty,0\right) + B\!\left(t,\tau;-\mu,0,0,-\infty,0\right) \Lambda}}
\rlap{\ ,}
\end{eqnarray*}
the price of the Receiver Swaption can be approximated by
\begin{multline*}
SR = L_M \frac{1}{\pi} \int_{-\infty}^{+\infty} dp \,
\frac{\sin(p(\alpha_\Lambda - \beta_\Lambda \Lambda))}{p} \left[ \int_t^T \ud \tau \, h(\tau) e^{A\!\left(t,\tau;-\mu,0,0,-\infty,0\right)} \right.
\\
\left.
\mathbb{E}\!\left(e^{-i p \widetilde{L}_t + i p \beta_\Lambda \lambda_t - \mu \widetilde{L}_t + B\!\left(t,\tau;-\mu,0,0,-\infty,0\right)\lambda_t }\right)
- \mathbb{E}\!\left( e^{-i p \widetilde{L}_t + i p \beta_\Lambda \lambda_t} \mathbbm{1}_{Q_t=0}\right) \right] ZC(0,t)
\end{multline*}
with
\begin{multline*}
\mathbb{E}\!\left(e^{-i p \widetilde{L}_t + i p \beta_\Lambda \lambda_t - \mu \widetilde{L}_t + B\!\left(t,\tau;-\mu,0,0,-\infty,0\right)\lambda_t }\right) =
\\
e^{A\!\left[0,t;-i p - \mu,0,i p \beta_\Lambda + B\!\left(t,\tau;-\mu,0,0,-\infty,0\right),0 \right] + B\!\left[0,t;-i p - \mu,0,i p + B\!\left(t,\tau;-\mu,0,0,-\infty,0\right) \beta_\Lambda ,0\right] \lambda_0}
\end{multline*}
and
\begin{equation*}
\mathbb{E}\!\left( e^{-i p \widetilde{L}_t + i p \beta_\Lambda \lambda_t} \mathbbm{1}_{Q_t=0} \right)
=
e^{A(0,t;-i p,0,i p \beta_\Lambda,-\infty,0) + B(0,t;-i p,0,i p \beta_\Lambda,-\infty,0) \lambda_0}
\rlap{\ .}
\end{equation*}

The price of a Payer Swaption is obtained by Call-Put parity:
\begin{equation*}
SP = L_M \left[ 1 - \int_t^T \ud \tau \, h(\tau) \mathbb{E}\!\left( e^{-\frac{\widetilde{L}_{\tau}}{L_M}} \mathbbm{1}_{Q_\tau=0} \right) \right] ZC(0,t) - SR
\end{equation*}
with
\begin{equation*}
\mathbb{E}\!\left( e^{-\frac{\widetilde{L}_{\tau}}{L_M}} \mathbbm{1}_{Q_\tau=0} \right)
=
e^{A\!\left(0,\tau;-\mu,0,0,-\infty,0\right) + B\!\left(0,\tau;-\mu,0,0,-\infty,0\right) \lambda_0}
\rlap{\ .}
\end{equation*}
\end{proposition}

As Call-Put parity holds for Receiver and Payer Swaptions, we consider here only a Receiver swaption.
The price of a Payer Swaption is given in term of the Receiver price as
\begin{equation*}
SP = \mathbb{E}_0\!\left(  CL_{CDS} - FL_{CDS}  + L_t \right) ZC(0,t) - SR
\rlap{\ .}
\end{equation*}

\begin{proof}
First, replacing $N$ by $L/l_1$, which is an approximation if the recovery rate is stochastic, the option holds on the quantity
\begin{equation*}
\mathbb{E}_t( FL_{CDS}-CL_{CDS}) - L_t =
\frac{K}{l_1} \int_t^T \ud \tau \left[ L_M - \mathbb{E}_t(L_{\tau}) \right] ZC(t,\tau)
- \int_t^T d \mathbb{E}_t(L_{\tau}) ZC(t,\tau) -L_t
\rlap{\ .}
\end{equation*}
Integrating by parts, this can be rephrased as
\begin{equation}
\mathbb{E}_t( FL_{CDS}-CL_{CDS}) - L_t = \int_t^T \ud \tau \, h(\tau) \left[ L_M - \mathbb{E}_t(L_{\tau}) \right] - L_M
\label{swaption1}
\end{equation}
where
\begin{equation*}
h(\tau) = \left[ \delta(\tau-T) + \frac{K}{l_1} + r(\tau) \right] ZC(t,\tau)
\rlap{\ .}
\end{equation*}
$r(\tau)$ is the instantaneous interest rate:
\begin{equation*}
ZC(t,\tau) = e^{\int_t^\tau ds r(s) }
\rlap{\ .}
\end{equation*}
Using $L_\tau \simeq L_M \left( 1-e^{-\mu \widetilde{L}_\tau} \mathbbm{1}_{Q_\tau=0} \right)$, equation \eqref{swaption1} gives
\begin{equation}
\label{swaption2}
\left[ \mathbb{E}_t( FL_{CDS}-CL_{CDS}) - L_t \right]^+ = L_M \left[ \int_t^T \ud \tau \, h(\tau) \mathbb{E}_t\!\left( e^{-\mu \widetilde{L}_{\tau}} \right) - 1 \right]^+
\rlap{\ .}
\end{equation}

The expected value at date $t$ is conditional to $\lambda_t$,  $\widetilde{L}_t$ and $Q_t$. According to theorem \ref{thm}, it is
\begin{equation}
 \mathbb{E}_t\!\left( e^{-\mu \widetilde{L}_{\tau}} \mathbbm{1}_{Q_\tau=0} \right) = e^{A\!\left(t,\tau;-\mu,0,0,-\infty,0\right) + B\!\left(t,\tau;-\mu,0,0,-\infty,0\right) \lambda_t - \mu \widetilde{L}_t }  \mathbbm{1}_{Q_t=0}
 \label{swaptionThm1}
\rlap{\ .}
\end{equation}
If $Q_t=0$ and for a given value of $\lambda_t$, the quantity inside the bracket in equation \eqref{swaption2} decreases with $\widetilde{L}_t$. There is a strike value $L_\star(\lambda_t)$ for which this quantity becomes negative and the payoff vanishes:
\begin{equation}
L_\star(\lambda_t) = \frac{1}{\mu} \ln\!\left( \int_t^T \ud \tau \, h(\tau) e^{A\!\left(t,\tau;-\mu,0,0,-\infty,0\right) + B\!\left(t,\tau;-\mu,0,0,-\infty,0\right) \lambda_t} \right)
\label{swaptionLstar}
\rlap{\ .}
\end{equation}
If the interest rate is not too negative, the integral is always positive. Depending of the strike $K$ and the value of $\lambda_t$ it may give a negative $L_\star(\lambda_t)$; in this case all values of $\widetilde{L}_t$ will contribute to the final payoff.

In particular if $L_\star(0) \leq 0$, and also if $Q_t>0$, the Receiver swaption is never exercised and its price is zero:
\begin{equation}
\label{ReceiverExercised}
SR = 0
\rlap{\ .}
\end{equation}
In the following we consider the nontrivial case where $L_\star(0) > 0$.

We make an approximation at this point: $L_\star(\lambda_t)$ is approximated by an affine function in $\lambda_t$. Starting from a typical value for $\lambda_t$ that we denote by $\Lambda$, equation \eqref{swaptionLstar} can be rewritten as
\begin{multline}
L_\star(\lambda_t) =\frac{1}{\mu} \ln\!\left( \int_t^T \ud \tau \, h(\tau) e^{A\!\left(t,\tau;-\mu,0,0,-\infty,0\right) + B\!\left(t,\tau;-\mu,0,0,-\infty,0\right) \Lambda} \right)
\\
+ \frac{1}{\mu} \ln\!\left( \int_t^T \ud \tau \, g_\Lambda(\tau) e^{B\!\left(t,\tau;-\mu,0,0,-\infty,0\right) (\lambda_t-\Lambda)} \right)
\label{swaptionLstar2}
\end{multline}
where $g_\Lambda(\tau)$ is a measure on time given by
\begin{equation*}
g(\tau) = \frac{h(\tau) e^{A\!\left(t,\tau;-\mu,0,0,-\infty,0\right) + B\!\left(t,\tau;-\mu,0,0,-\infty,0\right) \Lambda}}{\displaystyle \int_t^T \ud \tau' \, h(\tau') e^{A\!\left(t,\tau',-\mu,0,0\right) + B\!\left(t,\tau',-\mu,0,0\right) \Lambda}}
\rlap{\ .}
\end{equation*}
We finally linearize $L_\star(\lambda_t)$ as
\begin{equation*}
L_\star(\lambda_t) = \alpha_\Lambda + \beta_\Lambda (\lambda_t-\Lambda) + O\left( (\lambda_t-\Lambda)^2 \right)
\end{equation*}
where
\begin{eqnarray*}
\alpha_\Lambda &=& \frac{1}{\mu} \ln\!\left( \int_t^T \ud \tau \, h(\tau) e^{A\!\left(t,\tau;-\mu,0,0,-\infty,0\right) + B\!\left(t,\tau;-\mu,0,0,-\infty,0\right) \Lambda}\right)
\\
\beta_\Lambda &=& \frac{1}{\mu} \int_t^T \ud \tau g_\Lambda(\tau) B\!\left(t,\tau;-\mu,0,0,-\infty,0\right)
\rlap{\ .}
\end{eqnarray*}
To support this approximation, we can notice first that for short tenors, the Dirac function dominates in $h(\tau)$, which means that the measure $g_\Lambda(\tau)$ reduces to a Dirac delta $\delta(\tau-T)$. In this limit case, equation \eqref{swaptionLstar2} shows that the approximation is exact. At the opposite limit, for tenors which are long compared to the characteristic time $1/\kappa$, $B(t,\tau,\cdots)$ is almost everywhere equal to its limit value $-B_-(\cdots)$. The exponential can be factorized outside the integral, and the approximation becomes exact in this limit. Between these two cases (short and large tenors), the approximation can be controlled. The first term in the expansion in $\lambda_t - \Lambda$ which is neglected is the quadratic term. Its coefficient is the variance of $B(t,\tau,\cdots)$ under the measure $g_\Lambda(\tau)$. Note that errors on $L_\star(\lambda_t)$ have no importance where the joint density of $\lambda_t$ and $\widetilde{L}_t$ is negligible. This is especially true for negative $L_\star(\lambda_t)$ or negative $\lambda_t$ where the density is zero.

The price of the swaption is given by the risk-neutral expected value: the conditional payoff must be integrated over $\lambda_t$, $\widetilde{L}_t$ and $Q_t$ using the joint distribution of these two stochastic variables. In order to do this, we rewrite the payoff \eqref{swaption2} using the boundary $L_\star(\lambda_t)$ as
\begin{equation}
\left[ \mathbb{E}_t( FL_{CDS}-CL_{CDS}) - L_t \right]^+ =
L_M \left( \int_t^T \ud \tau \, h(\tau) \mathbb{E}_t\!\left( e^{-\mu \widetilde{L}_{\tau}} \mathbbm{1}_{Q_\tau=0} \right) \mathbbm{1}_{\widetilde{L}_t < L_\star(\lambda_t)}
-  \mathbbm{1}_{Q_t=0} \mathbbm{1}_{\widetilde{L}_t \leq L_\star(\lambda_t)} \right)
\label{swaption2terms}
\rlap{\ .}
\end{equation}
The expected value at time 0 of this expression will be computed by Fourier transform in $\widetilde{L}_t$. Both terms vanish at $+\infty$ but are constant or diverge in $-\infty$. As $\widetilde{L}_t$ has zero probability of being negative, a cut-off can safely be set at a negative value. Therefore
\begin{equation*}
\mathbbm{1}_{\widetilde{L}_t < L_\star(\lambda_t)}
\end{equation*}
can be replaced in both terms by
\begin{equation*}
\mathbbm{1}_{L_\star(\lambda_t) - M \leq \widetilde{L}_t \leq L_\star(\lambda_t)}
\end{equation*}
for any $M > L_\star(0)$ without changing their expected value.

The Fourier transform in $\widetilde{L}_t$ of this function is
\begin{equation*}
\int_{-\infty}^{+\infty} \ud \widetilde{L}_t \, e^{i p \widetilde{L}_t} \mathbbm{1}_{L_\star(\lambda_t) - M \leq \widetilde{L}_t \leq L_\star(\lambda_t)}
=
\frac{1-e^{-i p M}}{i p } e^{i p L_\star(\lambda_t)}
\rlap{\ .}
\end{equation*}
The inverse Fourier transform is the original function
\begin{equation}
\mathbbm{1}_{L_\star(\lambda_t) - M \leq \widetilde{L}_t \leq L_\star(\lambda_t)} =
\frac{1}{2\pi} \int_{-\infty}^{+\infty} \ud p \, e^{-i p \widetilde{L}_t} \frac{1-e^{-i p M}}{i p } e^{i p L_\star(\lambda_t)}
\label{inverseFourier}
\end{equation}
and gives the expected value today:
\begin{equation*}
\mathbb{E}\!\left( \mathbbm{1}_{Q_t=0} \mathbbm{1}_{\widetilde{L}_t \leq L_\star(\lambda_t)}\right) =
\frac{1}{2\pi} \int_{-\infty}^{+\infty} \ud p \, \frac{1-e^{-i p M}}{i p } \mathbb{E}\!\left(e^{-i p \widetilde{L}_t + i p L_\star(\lambda_t)}  \mathbbm{1}_{Q_t=0}\right)
\rlap{\ .}
\end{equation*}

Using the linear approximation $L_\star(\lambda_t) \simeq \alpha_\Lambda + \beta_\Lambda(\lambda_t-\Lambda)$, this expected value can be computed as a one-dimensional integral
\begin{equation*}
\mathbb{E}\!\left( \mathbbm{1}_{Q_t=0}\mathbbm{1}_{\widetilde{L}_t \leq L_\star(\lambda_t)}\right)
=
\frac{1}{2\pi} \int_{-\infty}^{+\infty} \ud p \,
\frac{1-e^{-i p M}}{i p} e^{i p (\alpha_\Lambda - \beta_\Lambda \Lambda)} \mathbb{E}\!\left( e^{-i p \widetilde{L}_t + i p \beta_\Lambda \lambda_t} \mathbbm{1}_{Q_t=0} \right)
\end{equation*}
where $\mathbb{E}\!\left( e^{-i p \widetilde{L}_t + i p \beta_\Lambda \lambda_t} \mathbbm{1}_{Q_t=0} \right)$ is computed using theorem \ref{thm}:
\begin{equation*}
\mathbb{E}\!\left( e^{-i p \widetilde{L}_t + i p \beta_\Lambda \lambda_t} \mathbbm{1}_{Q_t=0} \right)
=
e^{A(0,t;-i p,0,i p \beta_\Lambda,-\infty,0) + B(0,t;-i p,0,i p \beta_\Lambda,-\infty,0) \lambda_0}
\rlap{\ .}
\end{equation*}

As we are in the case $L_\star(0)>0$, this formula can be simplified if $M$ is set to
\begin{equation*}
M = 2 (\alpha_\Lambda - \beta_\Lambda \Lambda) \simeq 2 L_\star(0) > L_\star(0)
\rlap{\ .}
\end{equation*}
With this choice for $M$, the expected value is
\begin{equation*}
\mathbb{E}\!\left( \mathbbm{1}_{Q_t=0} \mathbbm{1}_{\widetilde{L}_t \leq L_\star(\lambda_t)}\right)
=
\frac{1}{\pi} \int_{-\infty}^{+\infty} \ud p \,
\frac{\sin(p(\alpha_\Lambda - \beta_\Lambda \Lambda))}{p} \mathbb{E}\!\left( e^{-i p \widetilde{L}_t + i p \beta_\Lambda \lambda_t} \mathbbm{1}_{Q_t=0} \right)
\rlap{\ .}
\end{equation*}

This gives the expected value of the second term in formula \eqref{swaption2terms}. The expected value of the first term
\begin{equation*}
\int_t^T \ud \tau \, h(\tau) \mathbb{E}_t\!\left( e^{-\mu \widetilde{L}_{\tau}} \mathbbm{1}_{Q_\tau=0} \right) \mathbbm{1}_{\widetilde{L}_t < L_\star(\lambda_t)}
\end{equation*}
is computed similarly. First, remind that the conditional expected value is given by equation \eqref{swaptionThm1}:
\begin{equation*}
\mathbb{E}_t\!\left( e^{-\mu \widetilde{L}_{\tau}} \mathbbm{1}_{Q_\tau=0} \right) =
e^{A\!\left(t,\tau;-\mu,0,0,-\infty,0\right) + B\!\left(t,\tau;-\mu,0,0,-\infty,0\right) \lambda_t - \mu \widetilde{L}_t} \mathbbm{1}_{Q_t=0}
\rlap{\ .}
\end{equation*}
$\mathbbm{1}_{\widetilde{L}_t < L_\star(\lambda_t)}$ is replaced by $\mathbbm{1}_{L_\star(\lambda_t) - M \leq \widetilde{L}_t \leq L_\star(\lambda_t)}$ which can be rewritten using equation \eqref{inverseFourier}. This gives
\begin{multline*}
\mathbb{E}_t\!\left( e^{-\mu \widetilde{L}_{\tau}} \mathbbm{1}_{Q_\tau=0} \right) \mathbbm{1}_{L_\star(\lambda_t) - M \leq \widetilde{L}_t \leq L_\star(\lambda_t)} =
\\
e^{A\!\left(t,\tau;-\mu,0,0,-\infty,0\right) + B\!\left(t,\tau;-\mu,0,0,-\infty,0\right) \lambda_t - \mu \widetilde{L}_t} \mathbbm{1}_{Q_t=0}
\frac{1}{2\pi} \int_{-\infty}^{+\infty} \ud p \, e^{-i p \widetilde{L}_t} \frac{1-e^{-i p M}}{i p } e^{i p L_\star(\lambda_t)}
\rlap{\ .}
\end{multline*}
The expected value at time 0 of this expression gives
\begin{multline*}
\mathbb{E}\!\left[ \mathbb{E}_t\!\left( e^{-\mu \widetilde{L}_{\tau}} \mathbbm{1}_{Q_\tau=0} \right) \mathbbm{1}_{\widetilde{L}_t < L_\star(\lambda_t)}\right] =
\\
 e^{A\!\left(t,\tau;-\mu,0,0,-\infty,0\right)} \frac{1}{2\pi} \int_{-\infty}^{+\infty} \ud p \,
\frac{1-e^{-i p M}}{i p } \mathbb{E}\!\left(e^{-i p \widetilde{L}_t + i p L_\star(\lambda_t) - \mu \widetilde{L}_t + B\!\left(t,\tau;-\mu,0,0,-\infty,0\right)\lambda_t } \mathbbm{1}_{Q_t=0} \right)
\rlap{\ .}
\end{multline*}

Using again the linear approximation $L_\star(\lambda_t) \simeq \alpha_\Lambda + \beta_\Lambda(\lambda_t-\Lambda)$, this expression can also be evaluated as a one-dimensional integral
\begin{multline*}
\mathbb{E}\!\left[ \mathbb{E}_t\!\left( e^{-\mu \widetilde{L}_{\tau}} \mathbbm{1}_{Q_\tau=0} \right) \mathbbm{1}_{\widetilde{L}_t < L_\star(\lambda_t)}\right] =
e^{A\!\left(t,\tau;-\mu,0,0,-\infty,0\right)}
\\
\frac{1}{2\pi} \int_{-\infty}^{+\infty} \ud p \,
\frac{1-e^{-i p M}}{i p }
e^{i p (\alpha_\Lambda - \beta_\Lambda \Lambda)}
\mathbb{E}\!\left(e^{-i p \widetilde{L}_t + i p \beta_\Lambda \lambda_t - \mu \widetilde{L}_t + B\!\left(t,\tau;-\mu,0,0,-\infty,0\right)\lambda_t } \mathbbm{1}_{Q_t=0} \right)
\end{multline*}
where $\mathbb{E}\!\left(e^{-i p \widetilde{L}_t + i p \beta_\Lambda \lambda_t - \mu \widetilde{L}_t + B\!\left(t,\tau;-\mu,0,0,-\infty,0\right)\lambda_t } \mathbbm{1}_{Q_t=0} \right)$ is given by theorem \ref{thm}:
\begin{multline*}
\mathbb{E}\!\left(e^{-i p \widetilde{L}_t + i p \beta_\Lambda \lambda_t - \mu \widetilde{L}_t + B\!\left(t,\tau;-\mu,0,0,-\infty,0\right)\lambda_t } \mathbbm{1}_{Q_t=0} \right) =
\\
e^{A\!\left[0,t;-i p - \mu,0,i p \beta_\Lambda + B\!\left(t,\tau;-\mu,0,0,-\infty,0\right),0 \right] + B\!\left[0,t;-i p - \mu,0,i p + B\!\left(t,\tau;-\mu,0,0,-\infty,0\right) \beta_\Lambda ,0\right] \lambda_0}
\rlap{\ .}
\end{multline*}

With the choice $M = 2(\alpha_\Lambda - \beta_\Lambda \Lambda)$ the integral is
\begin{multline*}
\mathbb{E}\!\left[ \mathbb{E}_t\!\left( e^{-\mu \widetilde{L}_{\tau}} \mathbbm{1}_{Q_\tau=0} \right) \mathbbm{1}_{\widetilde{L}_t < L_\star(\lambda_t)}\right] =
e^{A\!\left(t,\tau;-\mu,0,0,-\infty,0\right)}
\\
\frac{1}{\pi} \int_{-\infty}^{+\infty} \ud p \,
\frac{\sin(p(\alpha_\Lambda - \beta_\Lambda \Lambda))}{p}
\mathbb{E}\!\left(e^{-i p \widetilde{L}_t + i p \beta_\Lambda \lambda_t - \mu \widetilde{L}_t + B\!\left(t,\tau;-\mu,0,0,-\infty,0\right)\lambda_t } \mathbbm{1}_{Q_t=0} \right)
\rlap{\ .}
\end{multline*}

The summation over coupon dates finishes the computation of the Swaption price.
\end{proof}

\section{Calibration}
\label{sec-calibration}

\subsection{Algorithm}

The model can be calibrated on many instruments, depending on what it is used for. As they are the most liquid products, the market spread of index CDS should be matched. Other relatively liquid instruments, at least on traded indices, are standard CDO tranches. One may therefore want to match CDO market quotes for all available maturities. As the model for $\widetilde{L}$ can be the sum of many basic building blocks with their own piecewise constant parameters sets, in principle many (non-arbitrageable) sets of price can be reached. However, introducing too many parameters make the calibration more difficult and more unstable. With too many parameters, the model can even be under-specified. So we prefer to stick with one single building block. Moreover, time-dependency is introduced in the minimal way of section \ref{timeDependency} through a piecewise affine change of time. We calibrate the model to match as closely as possible market quotes for CDO tranches (either one maturity or all maturities) while matching precisely quotes for the index CDS. (It would be useful to calibrate also on swaptions: CDS or CDO tranches are not really sensitive to some dynamical properties of the model.)

The calibration routine uses the following steps:
\begin{enumerate}
\item Choose values for parameters $\lambda_0$, $\lambda_\infty$, $\kappa$, $\sigma$, $n$, $\theta$, $\gamma$, $\alpha$ and $\beta$.
\item Conditionally to this set of parameters, calibrate the slopes $a_i$ of the time-dependence function $t=t(\tau)$ by a bootstrap method. One dichotomy or Newton algorithm for each CDS maturity is needed, which is fast because there is a closed formula for the expected loss. CDS prices are computed as explained in section \ref{sec-cds}.
\item Price CDO tranches (as in section \ref{sec-cdo}) for this set of parameters and time-dependency function and compute the sum of squared errors, relatively to the bid-ask spreads.
\item Loop on step 2 according to some optimization algorithm.
\end{enumerate}
We use differential evolution \cite{storn1995differential} algorithm to find a global minimum of the objective function.

\subsection{Numerical results}

We take as an example iTraxx Europe Series 9 Version 1 (RED Code: 2I666VAI6) data on September 30th, 2009. Quotes for the index, 9\%--12\% and 12\%--22\% tranches are spreads in bps. 0\%--3\%, 3\%--6\% and 6\%--9\% tranches have a running spread of 500bps and the quote are upfronts in percent.

\subsubsection{Global calibration}

Table \ref{table-calib-global} and figure \ref{fig-calib-global}\footnote{We plot the quantity $\frac{s T + U}{1+sT/2}$, where $s$ is the spread quote (tranches 4 and 5) or the running spread (5\% for tranches 1, 2 and 3), $T$ is the maturity in years and $U$ the upfront quote for tranches 1, 2 and 3. Forgetting about discount factors and with the hypothesis that on a given tranche the loss per unit of time is constant, we can write the upfront in term of the tranche loss $L$ and running spread $s$ as $U\sim L-sT\left(1-\frac{L}{2}\right)$. This gives an approximation of the tranche loss: $L \sim \frac{s T + U}{1+sT/2}$.} show market quotes and the result of a global calibration on 5Y, 7Y and 10Y instruments (maturing on June 20th, 2013, 2015 and 2018). We set the recovery rate to 40\%.

\newcommand{\withrunningspread}[2]{$\displaystyle\mathop{#1}_{\tiny #2}$}
\begin{table}[p]
\centering
\begin{tabular}{|c|ccc|ccc|}
\cline{2-7}
\multicolumn{1}{c|}{}
& \multicolumn{3}{c|}{Market}
& \multicolumn{3}{c|}{Model}
\\
\multicolumn{1}{c|}{}
& 5Y & 7Y & 10Y
& 5Y & 7Y & 10Y
\\\hline
Index &\withrunningspread{102.505\%}{165}
&\withrunningspread{103.487\%}{170}
&\withrunningspread{104.985\%}{175}
&102.505\%&103.487\%&104.985\%
\\ 
0\%--3\% &\withrunningspread{36.81\%\ (1.06)}{500}&\withrunningspread{45.81\%\ (1.06)}{500}&
\withrunningspread{52.50\%\ (1.06)}{500}
&36.10\%&45.52\%&52.94\%
\\ 
3\%--6\% &\withrunningspread{2.83\%\ (1.06)}{500}&\withrunningspread{8.06\%\ (1.12)}{500}&\withrunningspread{14.88\%\ (1.12)}{500}
&1.49\%&8.55\%&14.93\%
\\
6\%--9\% &\withrunningspread{-6.95\%\ (0.88)}{500}&\withrunningspread{-5.59\%\ (0.88)}{500}&\withrunningspread{-1.94\%\ (0.88)}{500}
&-7.44\%&-5.13\%&-2.33\%
\\
9\%--12\% &147.75\ (12)&196.13\ (12)&247.00\ (11.5)
&127.20&200.01&259.85
\\
12\%--22\% &58.75\ (8)&81.88\ (8)&98.88\ (8)
&54.92&80.48&105.07
\\ \hline
\end{tabular}
\caption{\textit{Market quotes for iTraxx Europe Series 9, index and standard tranches on September 30, 2009 and model results after a global calibration. Index quotes are given as 100\%$-$upfront, the three first tranches are quoted upfront in percent; the running spreads in bp are given in small figures below quotes. The last two tranches have quotes in spread, in bp. The numbers in parentheses are bid-ask spreads.}}
\label{table-calib-global}
\end{table}

\begin{figure}[p]
\centering
\includegraphics[width=\textwidth]{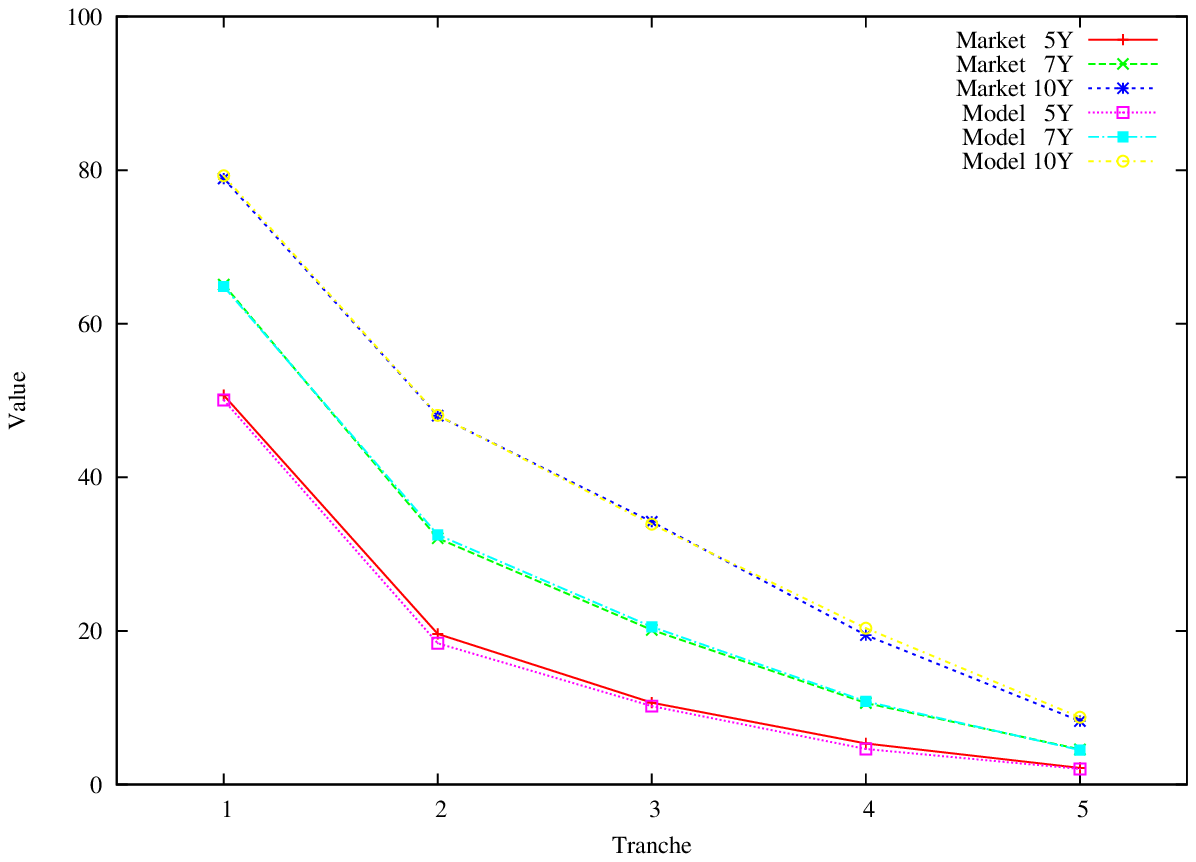}
\caption{\textit{Comparison of market and model values after a global calibration for iTraxx Europe Series 9 standard tranches on September 30, 2009. The value plotted is the quantity $\frac{s T + U}{1+sT/2}$, in percent, where $s$ is the spread quote in percent (tranches 4 and 5) or the running spread (5\% for tranches 1, 2 and 3), $T$ is the maturity in years and $U$ the upfront quote in percent for tranches 1, 2 and 3.}}
\label{fig-calib-global}
\end{figure}

\FloatBarrier

The parameters obtained by a global calibration on all maturities are the following:
\begin{center}
\begin{tabular}{cc}
$\lambda_0$ & 7.115 \\
$\lambda_\infty$ & 0.1846 \\
$\kappa$ & 4.303 \\
$\sigma$ & 0.9085 \\
\end{tabular}\hspace{2cm}
\begin{tabular}{cc}
$n$ & 9 \\
$\gamma$ & 0.08774 \\
$\theta$ & 5.704 \\
$\alpha$ & 1.377 $10^{-15}$ \\
\end{tabular}\hspace{2cm}
\begin{tabular}{cc}
$\beta$ & 0.005865 \\
$a_1$ & 0.766 \\
$a_2$ & 1.268 \\
$a_3$ & 1.111\rlap{\ .}\\
\end{tabular}
\end{center}
$a_1$, $a_2$ and $a_3$ are the slopes of the piecewise affine function $t=t(\tau)$ used to match index CDS quotes. Calibrating eight free parameters plus three slopes we manage to fit exactly the three index quotes and get a reasonable fit of the 15 CDO market quotes. Seven CDO quotes fall precisely within the bid-ask spread and the other are a little outside.

\subsubsection{Single maturity calibration}

In order to test the model ability to match all quotes for a given maturity, we keep five free parameters and freeze three quantities to the values obtained in the global calibration: $\frac{\lambda_\infty}{\kappa} = 0.04289$, $\frac{\sigma^2}{\kappa \lambda_\infty} = 0.5195$ and $\alpha=0$. We calibrate on CDO tranches of one single maturity but we still match the index quotes of the three maturities using the piecewise affine time change. The results are shown in table \ref{table-calib-single} and figure \ref{fig-calib-single}. All quotes are perfectly matched within the bid-ask spread. The model parameters for each maturity are the following:
\begin{center}
\begin{tabular}{cccc}
&5Y&7Y&10Y \\
$\lambda_0$ & 1.013 & 1.815 & 6.379\\
$\lambda_\infty$ & 0.01748 & 0.03061 & 0.3463\\
$\kappa$ & 0.4076 & 0.7135 & 8.075\\
$\sigma$ & 0.06084 & 0.1065 & 1.205\\
$n$ & 4 & 17 & 34\\
$\gamma$ & 0.1049 & 0.09124 & 0.07845\\
$\theta$ & 1.622 & 0.6226 & 3.131\\
$\alpha$ & 0 & 0 &0\\
$\beta$ & 0.004045 & 0.005947 & 0.006013\\
$a_1$ & 0.996 & 0.851 & 0.834\\
$a_2$ & 1.104 & 1.193 & 1.215\\
$a_3$ & 0.936 & 1.056 & 1.063\rlap{\ .}\\
\end{tabular}
\end{center}

\begin{table}[p]
\centering
\begin{tabular}{|c|ccc|ccc|}
\cline{2-7}
\multicolumn{1}{c|}{}
& \multicolumn{3}{c|}{Market}
& \multicolumn{3}{c|}{Model}
\\
\multicolumn{1}{c|}{}
& 5Y & 7Y & 10Y
& 5Y & 7Y & 10Y
\\\hline
Index &\withrunningspread{102.505\%}{165}
&\withrunningspread{103.487\%}{170}
&\withrunningspread{104.985\%}{175}
&102.505\%&103.487\%&104.985\%
\\ 
0\%--3\% &\withrunningspread{36.81\%\ (1.06)}{500}&\withrunningspread{45.81\%\ (1.06)}{500}&
\withrunningspread{52.50\%\ (1.06)}{500}
&36.78\%&45.80\%&52.52\%
\\ 
3\%--6\% &\withrunningspread{2.83\%\ (1.06)}{500}&\withrunningspread{8.06\%\ (1.12)}{500}&\withrunningspread{14.88\%\ (1.12)}{500}
&2.79\%&8.05\%&14.90\%
\\
6\%--9\% &\withrunningspread{-6.95\%\ (0.88)}{500}&\withrunningspread{-5.59\%\ (0.88)}{500}&\withrunningspread{-1.94\%\ (0.88)}{500}
&-6.98\%&-5.59\%&-1.92\%
\\
9\%--12\% &147.75\ (12)&196.13\ (12)&247.00\ (11.5)
&148.44&196.29&245.97
\\
12\%--22\% &58.75\ (8)&81.88\ (8)&98.88\ (8)
&57.61&81.40&100.31
\\ \hline
\end{tabular}
\caption{\textit{Market quotes for iTraxx Europe Series 9, index and standard tranches on September 30, 2009 and model results after one calibration for each maturity.  Index quotes are given as 100\%$-$upfront, the three first tranches are quoted upfront in percent; the running spreads in bp are given in small figures below quotes. The last two tranches have quotes in spread, in bp. The numbers in parentheses are bid-ask spreads.}}
\label{table-calib-single}
\end{table}

\begin{figure}[p]
\centering
\includegraphics[width=\textwidth]{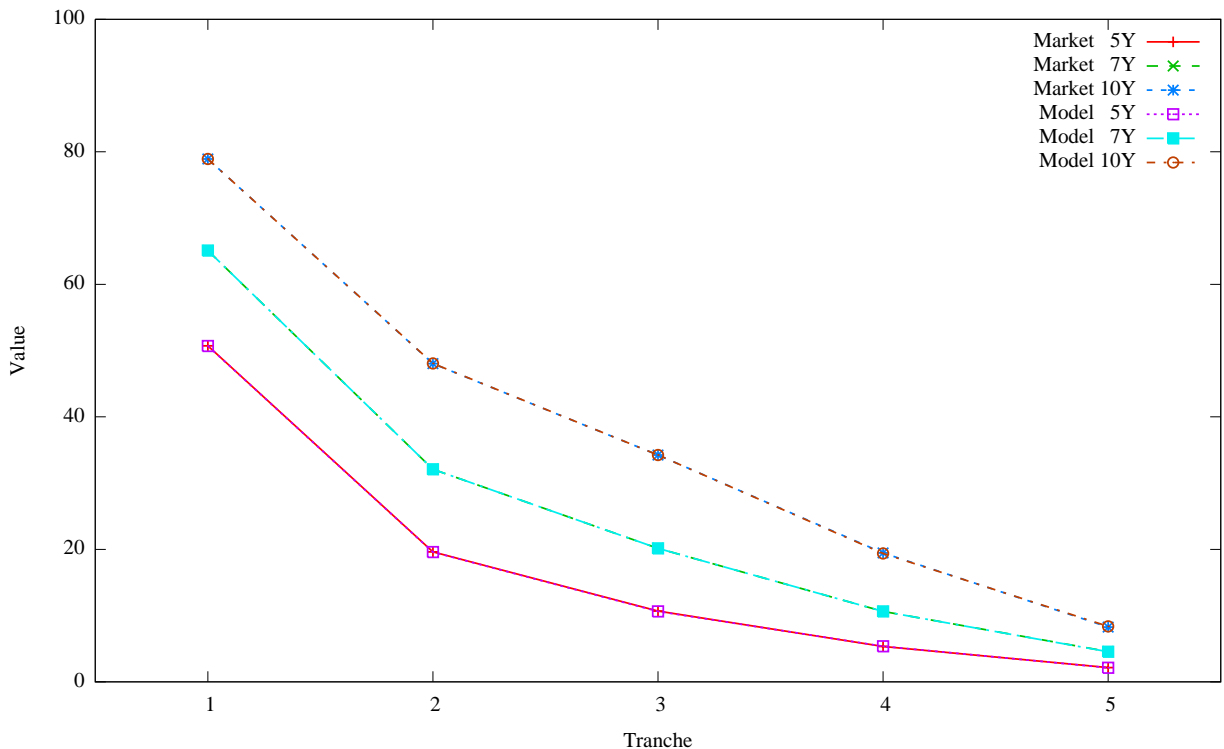}
\caption{\textit{Comparison of market and model values after one calibration for each maturity for iTraxx Europe Series 9 standard tranches on September 30, 2009. The value plotted is the quantity $\frac{s T + U}{1+sT/2}$, in percent, where $s$ is the spread quote in percent (tranches 4 and 5) or the running spread (5\% for tranches 1, 2 and 3), $T$ is the maturity in years and $U$ the upfront quote in percent for tranches 1, 2 and 3.}}
\label{fig-calib-single}
\end{figure}

\FloatBarrier

\section{Conclusion}
We have presented a model for credit index derivatives with the following properties: dynamic, with volatility and jumps on the loss intensity, allowing closed or semi-closed pricing of CDS, CDO tranches, N\textsuperscript{th}-to-default and options. It is also possible to price options on tranches or other exotic products since forward conditional distributions can be computed. 

In addition to default risk, our model handles volatility of the index spread and jump events, which model possible crises. Beyond pricing, this can be useful for risk management. Counterparty risk can also be included.

It is possible to match exactly index prices while matching reasonably CDO tranches quotes. If only one maturity is considered for CDOs, we obtained a perfect calibration within the bid-ask spread. It would be interesting to calibrate jointly on option prices to fix the remaining parameters (mean reversion and volatility in particular) and get robust sets of parameters.

It would be interesting to study our model in the spirit of \cite{cont2008dynamic} to see how it performs with respect to hedging of default risk and spread risk: it contains several properties that the authors find desirable, in particular the presence of jumps in the spread, not caused by the default of one of the constituents.




\pagebreak
\bibliographystyle{alpha}
\bibliography{credit}

\end{document}